\documentclass[11pt]{article}

\usepackage{amsmath}
\usepackage{amsthm}
\usepackage{thmtools}
\usepackage{xcolor}
\usepackage{amsfonts}
\usepackage{mathtools}
\usepackage{xspace}
\usepackage{array}
\usepackage{subcaption}
\usepackage{enumitem}
\usepackage{authblk}
\usepackage[margin=1in]{geometry}

\newtheorem{theorem}{Theorem}
\newtheorem{lemma}[theorem]{Lemma}
\newtheorem{definition}[theorem]{Definition}
\newtheorem{observation}[theorem]{Observation}

\newcommand{\Oh}{\mathcal{O}}

\newcommand{\etal}{et~al.\xspace}

\title{Computing Continuous Dynamic Time Warping of Time Series in Polynomial Time}
\author{Kevin Buchin$^1$, Andr\'e Nusser$^2$ and Sampson Wong$^3$}

\date{%
    \small
    $^1$Technical University of Dortmund, Germany, kevin.buchin@tu-dortmund.de\\%
    $^2$BARC, University of Copenhagen, Denmark, anusser@mpi-inf.mpg.de\footnote{Part of this research was conducted while the author was at Saarbrücken Graduate School of Computer Science and Max Planck Institute for Informatics. The author is supported by the VILLUM Foundation grant 16582.}\\%
    $^3$University of Sydney, Australia, swon7907@uni.sydney.edu.au\\%
}

\begin{document}

\maketitle

\begin{abstract}
Dynamic Time Warping is arguably the most popular similarity measure for time series, where we define a time series to be a one-dimensional polygonal curve. The drawback of Dynamic Time Warping is that it is sensitive to the sampling rate of the time series. The Fréchet distance is an alternative that has gained popularity, however, its drawback is that it is sensitive to outliers.

Continuous Dynamic Time Warping (CDTW) is a recently proposed alternative that does not exhibit the aforementioned drawbacks. CDTW combines the continuous nature of the Fréchet distance with the summation of Dynamic Time Warping, resulting in a similarity measure that is robust to sampling rate and to outliers.
In a recent experimental work of Brankovic et al., it was demonstrated that clustering under CDTW avoids the unwanted artifacts that appear when clustering under Dynamic Time Warping and under the Fréchet distance. Despite its advantages, the major shortcoming of CDTW is that there is no exact algorithm for computing CDTW, in polynomial time or otherwise.

In this work, we present the first exact algorithm for computing CDTW of one-dimensional curves. Our algorithm runs in time $\mathcal{O}(n^5)$ for a pair of one-dimensional curves, each with complexity at most $n$. In our algorithm, we propagate continuous functions in the dynamic program for CDTW, where the main difficulty lies in bounding the complexity of the functions. We believe that our result is an important first step towards CDTW becoming a practical similarity measure between curves.
\end{abstract}

\section{Introduction}
\label{sec:introduction}

Time series data arises from many sources, such as financial markets~\cite{taylor2008modelling}, seismology~\cite{yilmaz2001seismic}, electrocardiography~\cite{berkaya2018survey} and epidemiology~\cite{bhaskaran2013time}. Domain-specific questions can often be answered by analysing these time series. A common way of analysing time series is by finding similarities. Computing similarities is also a fundamental building block for other analyses, such as clustering, classification, or simplification. There are numerous similarity measures considered in literature~\cite{DBLP:journals/csur/AtluriKK18,cleasby2019using,DBLP:journals/csur/EslingA12,koenklaren,tao2021comparative,DBLP:journals/sigspatial/TooheyD15}, many of which are application dependent. 

Dynamic Time Warping (DTW) is arguably the most popular similarity measure for time series, and is widely used across multiple communities~\cite{DBLP:conf/compgeom/AgarwalFPY16,DBLP:journals/talg/GoldS18,muller2007dynamic,myers1980performance,sakoe1978dynamic,senin2008dynamic,tappert1990state,vintsyuk1968speech,DBLP:journals/datamine/WangMDTSK13}. Under DTW, a minimum cost discrete alignment is computed between a pair of time series. A discrete alignment is a sequence of pairs of points, subject to the following four conditions: \emph{(i)} the first pair is the first sample from both time series, \emph{(ii)} the last pair is the last sample from both time series, \emph{(iii)} each sample must appear in some pair in the alignment, and \emph{(iv)} the alignment must be a monotonically increasing sequence for both time series. The cost of a discrete alignment, under DTW, is the sum of the distances between aligned points. A drawback of a similarity measure with a discrete alignment is that it is sensitive to the sampling rates of the time series. As such, DTW is a poor measure of similarity between a time series with a high sampling rate and a time series with a low sampling rate.  For such cases, it is more appropriate to use a similarity measure with a continuous alignment. In Figure~\ref{fig:dtw_sampling_artifact}, we provide a visual comparison of a discrete alignment versus a continuous alignment, for time series with vastly different sampling rates.

\begin{figure}
	\begin{subfigure}[b]{0.47\textwidth}
    \centering
    \includegraphics[width=\textwidth]{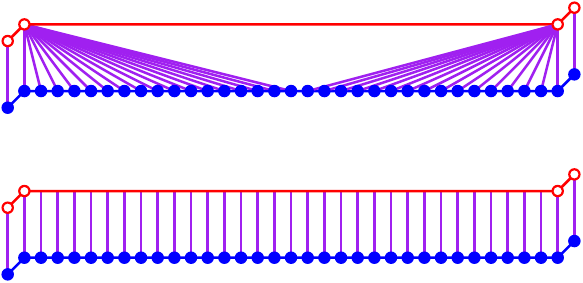}
    \caption{Top: The optimal alignment under a discrete similarity measure, e.g. DTW. 
    Bottom: The optimal alignment under a continuous similarity measure.}
    \label{fig:dtw_sampling_artifact}
	\end{subfigure}
	\hfill
	\begin{subfigure}[b]{0.47\textwidth}
    \centering
    \includegraphics[width=\textwidth]{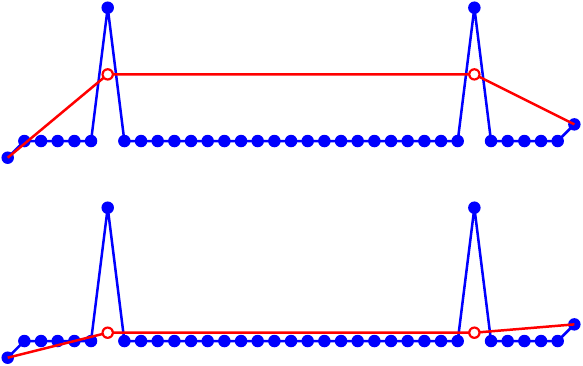}
    \caption{Top: The optimal simplification under a bottleneck measure, e.g., Fr\'echet distance. Bottom: The optimal simplification under a summation-based similarity measure.}
    \label{fig:frechet_sampling_artifact}
	\end{subfigure}
	\caption{Issues with discrete (left) and bottleneck (right) measures as opposed to continuous, summed measures.}
	\label{fig:issues}
\end{figure}

The Fr\'echet distance is a similarity measure that has gained popularity, especially in the theory community~\cite{DBLP:journals/ijcga/AltG95,DBLP:journals/dcg/BuchinBMM17,DBLP:conf/soda/DriemelKS16,DBLP:conf/icdar/SriraghavendraKB07}. To apply the Fr\'echet distance to a time series, we linearly interpolate between sampled points to obtain a continuous one-dimensional polygonal curve. Under the Fr\'echet distance, a minimum cost continuous alignment is computed between the pair of curves. A continuous alignment is a simultaneous traversal of the pair of curves that satisfies the same four conditions as previously stated for DTW. The cost of a continuous alignment, under the Fr\'echet distance, is the maximum distance between a pair of points in the alignment. The Fr\'echet distance is a bottleneck measure in that it only measures the maximum distance between aligned points. As a result, the drawback of the Fr\'echet distance is that it is sensitive to outliers. For such cases, a summation-based similarity measure is significantly more robust. In Figure~\ref{fig:frechet_sampling_artifact}, we illustrate a high complexity curve, and its low complexity ``simplification'' that is the most similar to the original curve, under either a bottleneck or summation similarity measure. The simplified curve under the Fr\'echet distance is sensitive to and drawn towards its outlier points.

Continuous Dynamic Time Warping (CDTW) is a recently proposed alternative that does not exhibit the aforementioned drawbacks. It obtains the best of both worlds by combining the continuous nature of the Fréchet distance with the summation of DTW. CDTW was first introduced by Buchin~\cite{buchin2007computability}, where it was referred to as the average Fr\'echet distance. CDTW has also been referred to as the summed, or integral, Fr\'echet distance. CDTW is similar to the Fr\'echet distance in that a minimum cost continuous alignment is computed between the pair of curves. The cost of a continuous alignment, under CDTW, is the integral of the distances between pairs of points in the alignment. We provide a formal definition in Section~\ref{sec:preliminaries}. Other definitions were also given under the name CDTW~\cite{DBLP:journals/jmiv/EfratFV07,DBLP:conf/cvpr/SerraB94}, see Section~\ref{sec:related_work}.


Compared to existing popular similarity measures, CDTW is robust to both the sampling rate of the time series and to its outliers. CDTW has been used in applications where this robustness is desirable. In Brakatsoulas~\etal~\cite{DBLP:conf/vldb/BrakatsoulasPSW05}, the authors applied CDTW to map-matching of vehicular location data. The authors highlight two common errors in real-life vehicular data, that is, measurement errors and sampling errors. Measurement errors result in outliers whereas sampling errors cause discrepancies in sampling rates between input curves. Their experiments show an improvement in map-matching when using CDTW instead of the Fr\'echet distance. In a recent paper, Brankovic~\etal~\cite{DBLP:conf/gis/BrankovicBKNPW20} applied CDTW to clustering of bird migration data and handwritten character data. The authors used $(k,\ell)$-center and medians clustering, where each of the $k$ clusters has a (representative) center curve of complexity at most $\ell$. Low complexity center curves are used to avoid overfitting. Compared to DTW and the Fr\'echet distance, Brankovic~\etal~\cite{DBLP:conf/gis/BrankovicBKNPW20} demonstrated that clustering under CDTW produced centers that were more visually similar to the expected center curve. Under DTW, the clustering quality deteriorated for small values of $\ell$, whereas under the Fr\'echet distance, the clustering quality deteriorated in the presence of outliers. 

\begin{figure}
    \centering
    \includegraphics[width=0.32\textwidth, trim=490 10 490 10, clip=true]{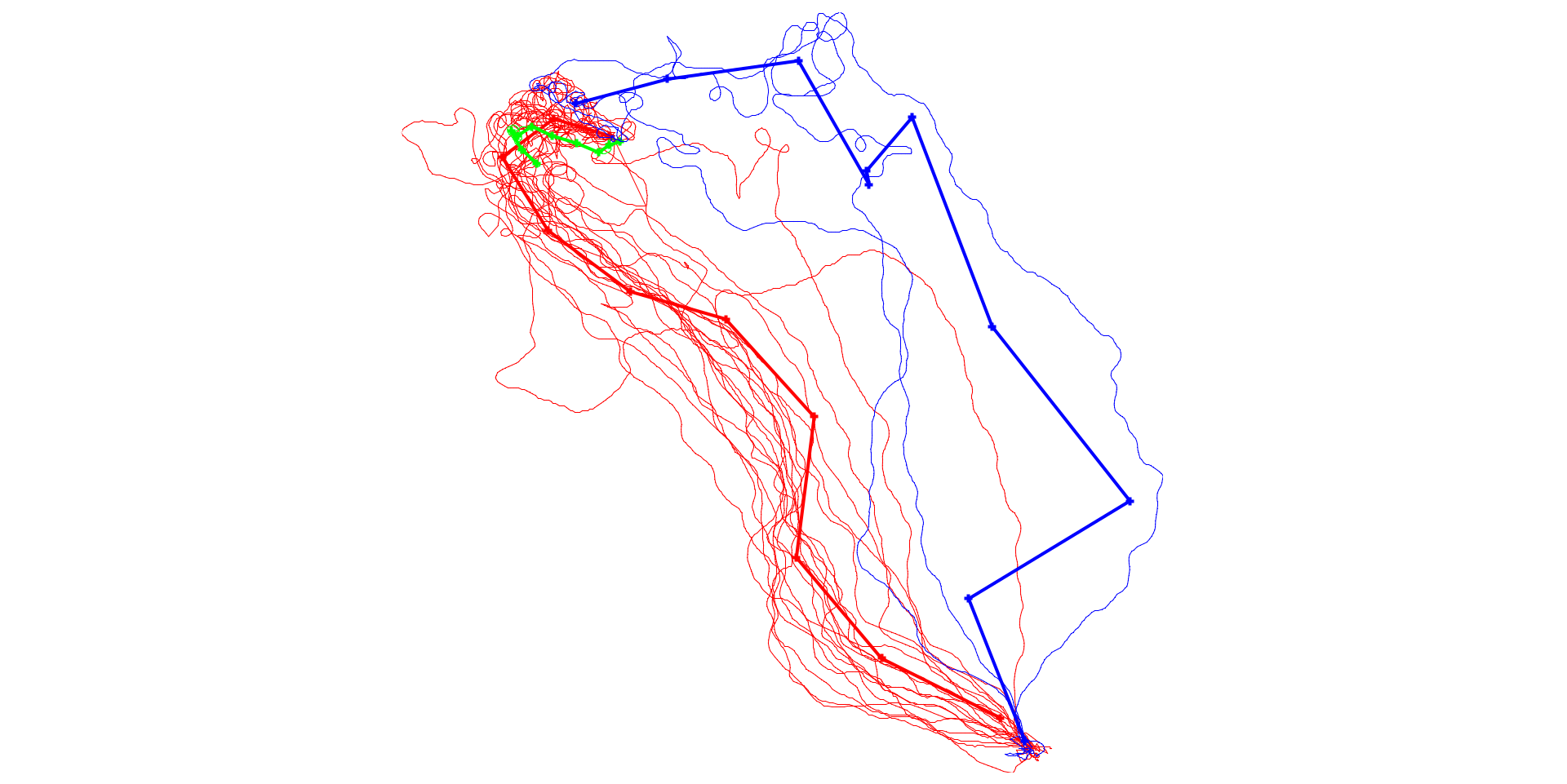}
    \includegraphics[width=0.32\textwidth, trim=490 10 490 10, clip=true]{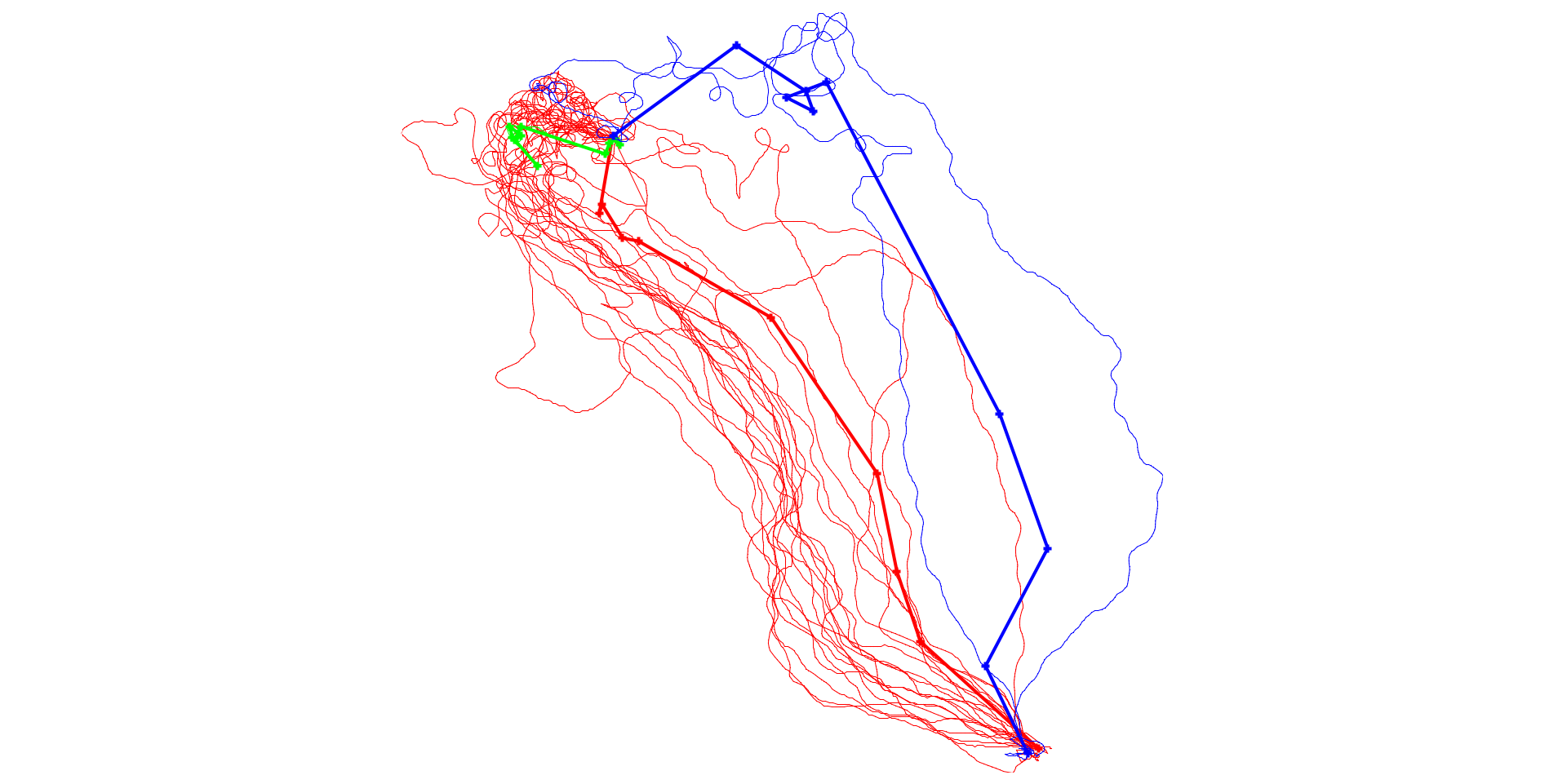}
    \includegraphics[width=0.32\textwidth, trim=490 10 490 10, clip=true]{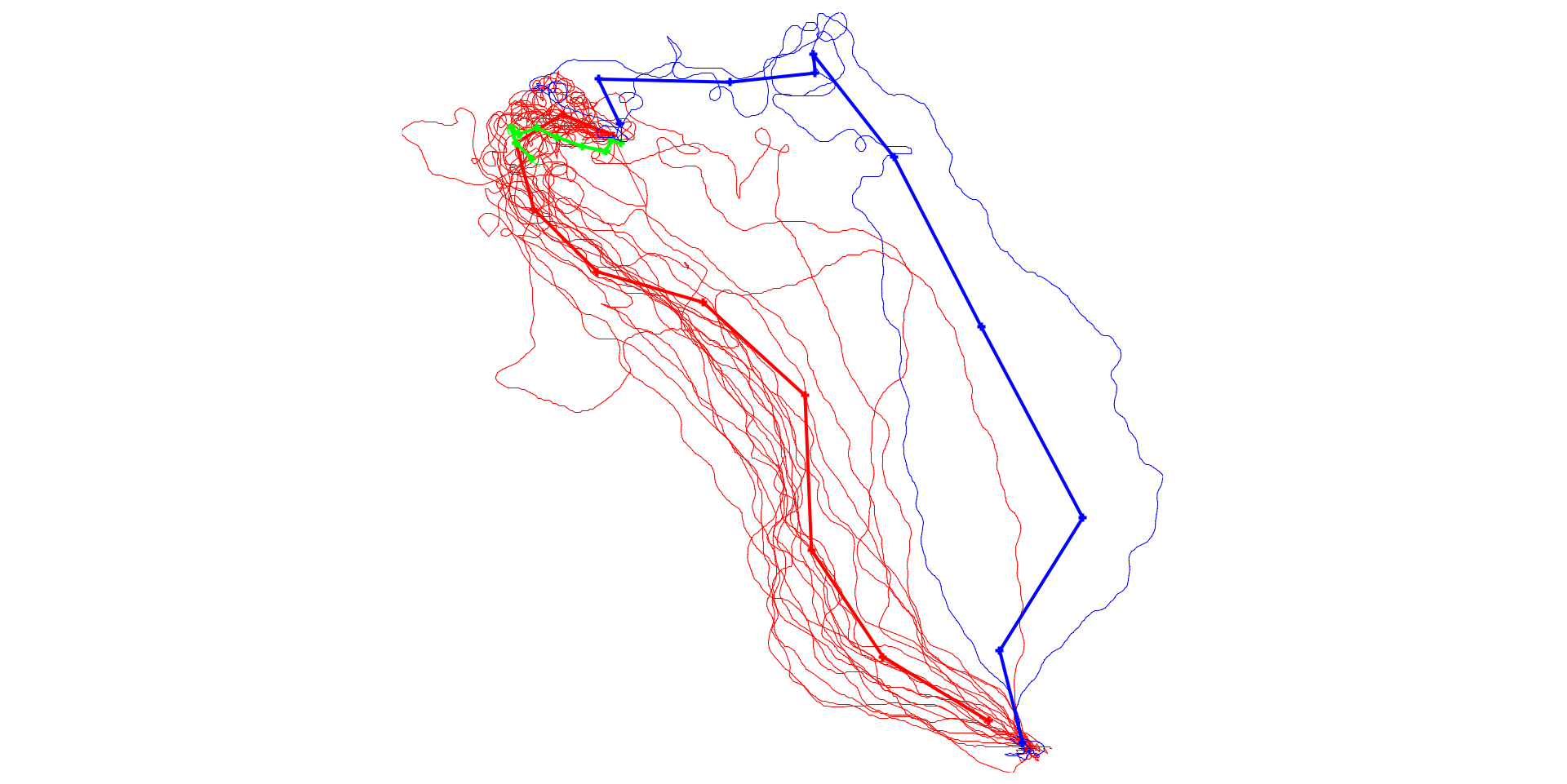}
	\caption{Clustering of the c17 pigeon's trajectories under the DTW (left), Fr\'echet (middle), and CDTW (right) distances. Figures were provided by the authors of~\cite{DBLP:conf/gis/BrankovicBKNPW20}.}
    \label{fig:c17}
\end{figure}

Brankovic~\etal's~\cite{DBLP:conf/gis/BrankovicBKNPW20} clustering of a pigeon data set~\cite{meade2005homing} is shown in Figure~\ref{fig:c17}. The Fr\'echet distance is paired with the center objective, whereas DTW and CDTW are paired with the medians objective. Under DTW (left), the discretisation artifacts are visible. The blue center curve is jagged and visually dissimilar to its associated input curves. Under the Fr\'echet distance (middle), the shortcoming of the bottleneck measure and objective is visible. The red center curve fails to capture the shape of its associated input curves, in particular, it misses the top-left ``hook'' appearing in its associated curves. Under CDTW (right), the center curves are smooth and visually similar to their associated curves.

Despite its advantages, the shortcoming of CDTW is that there is no exact algorithm for computing it, in polynomial time or otherwise. Heuristics were used to compute CDTW in the map-matching~\cite{DBLP:conf/vldb/BrakatsoulasPSW05} and clustering~\cite{DBLP:conf/gis/BrankovicBKNPW20} experiments. Maheshwari~\etal~\cite{DBLP:journals/comgeo/MaheshwariSS18} provided a $(1+\varepsilon)$-approximation algorithm for CDTW in $\Oh(\zeta^4 n^4 / \varepsilon^2)$ time, for curves of complexity $n$ and spread $\zeta$, where the spread is the ratio between the maximum and minimum interpoint distances. Existing heuristic and approximation methods~\cite{DBLP:conf/vldb/BrakatsoulasPSW05,DBLP:conf/gis/BrankovicBKNPW20,DBLP:journals/comgeo/MaheshwariSS18} use a sampled grid on top of the dynamic program for CDTW, introducing an inherent error that depends on the fineness of the sampled grid, which is reflected in the dependency on $\zeta$ in~\cite{DBLP:journals/comgeo/MaheshwariSS18}.

In this work, we present the first exact algorithm for computing CDTW for one-dimensional curves. Our algorithm runs in time $\Oh(n^5)$ for a pair of one-dimensional curves, each with complexity at most~$n$. Unlike previous approaches, we avoid using a sampled grid and instead devise a propagation method that solves the dynamic program for CDTW exactly. In our propagation method, the main difficulty lies in bounding the total complexity of our propagated functions. Showing that CDTW can be computed in polynomial time fosters hope for faster polynomial time algorithms, which would add CDTW to the list of practical similarity measures for curves. 

\subsection{Related work}
\label{sec:related_work}

Algorithms for computing popular similarity measures, such as DTW and the Fr\'echet distance, are well studied. Vintsyuk~\cite{vintsyuk1968speech} proposed Dynamic Time Warping as a similarity measure for time series, and provided a simple dynamic programming algorithm for computing the DTW distance that runs in $\Oh(n^2)$ time, see also \cite{DBLP:conf/kdd/BerndtC94}. Gold and Sharir~\cite{DBLP:journals/talg/GoldS18} improved the upper bound for computing DTW to $\Oh(n^2/ \log \log n)$. For the Fr\'echet distance, Alt and Godau~\cite{DBLP:journals/ijcga/AltG95} proposed an $\Oh(n^2 \log n)$ time algorithm for computing the Fr\'echet distance between a pair of curves. Buchin~\etal~\cite{DBLP:journals/dcg/BuchinBMM17} improved the upper bound for computing the Fr\'echet distance to $\Oh(n^2 \sqrt{\log n} (\log \log n)^{3/2})$. Assuming SETH, it has been shown that there is no strongly subquadratic time algorithm for computing the Fr\'echet distance or DTW~\cite{DBLP:conf/focs/AbboudBW15,DBLP:conf/focs/Bringmann14,DBLP:conf/focs/BringmannK15,DBLP:journals/jocg/BringmannM16,DBLP:conf/soda/BuchinOS19}.


Our definition of CDTW was originally proposed by Buchin~\cite{buchin2007computability}, and has since been used in several experimental works~\cite{DBLP:conf/vldb/BrakatsoulasPSW05,DBLP:conf/gis/BrankovicBKNPW20}. We give Buchin's~\cite{buchin2007computability} definition formally in Section~\ref{sec:preliminaries}. Other definitions under the name CDTW have also been considered. We briefly describe the main difference between these definitions and the one used in this paper. 

To the best of our knowledge, the first continuous version of DTW was by Serra and Berthod~\cite{DBLP:conf/cvpr/SerraB94}. The same definition was later used by Munich and Perona~\cite{DBLP:conf/iccv/MunichP99}. Although a continuous curve is used in their definition, the cost of the matching is still a discrete summation of distances to sampled points. Our definition uses a continuous summation (i.e. integration) of distances between all points on the curves, and therefore, is more robust to discrepancies in sampling rate. Efrat~\etal~\cite{DBLP:journals/jmiv/EfratFV07} proposed a continuous version of DTW that uses integration. However, their integral is defined in a significantly different way to ours. Their formulation minimises the change of the alignment and not the distance between aligned points. Thus, their measure is translational invariant and designed to compare the shapes of curves irrespective of their absolute positions in space.

\section{Preliminaries}
\label{sec:preliminaries}

We use $[n]$ to denote the set $\{1, \dots, n\}$.
To continuously measure the similarity of time series, we linearly interpolate between sampled points to obtain a one-dimensional polygonal curve. A one-dimensional polygonal curve $P$ of complexity~$n$ is given by a sequence of vertices, $p_1, \ldots, p_n \in \mathbb R$, connected in order by line segments. Furthermore, let $||\cdot||$ be the norm in the one-dimensional space $\mathbb R$.
In higher dimensions, the Euclidean $\mathcal L_2$ norm is the most commonly used norm, but other norms such as $\mathcal L_1$ and $\mathcal L_\infty$ may be used. 

Consider a pair of one-dimensional polygonal curves $P = p_1, \ldots, p_n$ and $Q = q_1, \ldots, q_m$. Let $\Delta(n,m)$ be the set of all sequences of pairs of integers $(x_1,y_1), \ldots, (x_k,y_k)$ satisfying $(x_1, y_1) = (1,1)$, $(x_k, y_k) = (n,m)$ and $(x_{i+1},y_{i+1}) \in \{(x_i+1,y_i), (x_i,y_i+1), (x_i+1,y_i+1)\}$. 
The DTW distance between $P$ and $Q$ is defined as
\[
    d_{DTW}(P,Q) = \min_{\alpha \in \Delta(n,m)} \sum_{(x,y) \in \alpha} ||p_x - q_y||.
\]
The discrete Fr\'echet distance between $P$ and $Q$ is defined as
\[
    d_{dF}(P,Q) = \min_{\alpha \in \Delta(n,m)} \max_{(x,y) \in \alpha} ||p_x - q_y||.
\]

Let $p$ and $q$ be the total arc lengths of $P$ and $Q$ respectively. Define the parametrised curve $\{P(z):z \in [0,p]\}$ to be the one-dimensional curve~$P$ parametrised by its arc length. In other words, $P(z)$ is a piecewise linear function so that the arc length of the subcurve from $P(0)$ to $P(z)$ is~$z$. Define $\{Q(z): z \in [0,q]\}$ analogously. Let $\Gamma(p)$ be the set of all continuous and non-decreasing functions $\alpha: [0,1] \to [0,p]$ satisfying $\alpha(0) = 0$ and $\alpha(1) = p$. Let $\Gamma(p,q) = \Gamma(p) \times \Gamma(q)$.
The~continuous Fr\'echet distance between $P$ and $Q$ is defined as
\[
    d_{F}(P,Q) = \inf_{(\alpha,\beta) \in \Gamma(p,q)} \max_{z \in [0,1]} ||P(\alpha(z)) - Q(\beta(z))||,
\]
The CDTW distance between $P$ and $Q$ is defined as
\[
    d_{CDTW}(P,Q) = \inf_{(\alpha,\beta) \in \Gamma(p,q)} \int_0^1 ||P(\alpha(z)) - Q(\beta(z))|| \cdot ||\alpha'(z) + \beta'(z)|| \cdot dz.
\]

For the definition of CDTW, we additionally require that $\alpha$ and $\beta$ are differentiable. The original intuition behind $d_{CDTW}(P,Q)$ is that it is a line integral in the parameter space, which we will define in Section~\ref{subsec:parameter_space}. The term $||\alpha'(z) + \beta'(z)||$ implies that we are using the $\mathcal L_1$ metric in the parameter space, but other norms have also been considered~\cite{koenklaren,DBLP:journals/comgeo/MaheshwariSS18}. 

\subsection{Parameter space under CDTW}
\label{subsec:parameter_space}

The parameter space under CDTW is analogous to the free space diagram under the continuous Fr\'echet distance. Similar to previous work~\cite{buchin2007computability,koenklaren,DBLP:journals/comgeo/MaheshwariSS18}, we transform the problem of computing CDTW into the problem of computing a line integral in the parameter space. 

Recall that the total arc lengths of $P$ and $Q$ are $p$ and $q$ respectively. The parameter space is defined to be the rectangular region $R = [0,p] \times [0,q]$ in $\mathbb R^2$. The region is imbued with a metric $||\cdot||_R$. The $\mathcal L_1$, $\mathcal L_2$ and $\mathcal L_\infty$ norms have all been considered, but $\mathcal L_1$ is the preferred metric as it is the easiest to work with~\cite{koenklaren,DBLP:journals/comgeo/MaheshwariSS18}. At every point $(x,y) \in R$ we define the height of the point to be $h(x,y) = ||P(x) - Q(y)||$.

Next, we provide the line integral formulation of $d_{CDTW}$, which is the original motivation behind its definition. To make our line integral easier to work with, we parametrise our line integral path $\gamma$ in terms of its $\mathcal L_1$ arc length in~$R$. The following lemma is a consequence of Section 6.2 in~\cite{buchin2007computability}. We provide a proof sketch of the result for the sake of self-containment.

\begin{restatable}{lem}{cdtwformula}
\label{lem:cdtw_formula}
\[
    d_{CDTW}(P,Q) = \inf_{\gamma \in \Psi(p,q)} \int_0^{p+q} h(\gamma(z)) \cdot dz,
\]
where $\Psi(p,q)$ is the set of all functions $\gamma:[0,p+q] \to R$ satisfying $\gamma(0) = (0,0)$, $\gamma(p+q) = (p,q)$, $\gamma$ is differentiable and non-decreasing in both $x$- and $y$-coordinates, and $||\gamma'(z)||_R=1$.
\end{restatable}

\begin{proof}[Proof (Sketch)] 
The full proof can be found in Appendix~\ref{apx:proof_of_cdtw_formula}. We summarise the main steps. Recall that the formula for CDTW is $\inf_{(\alpha,\beta) \in \Gamma(p,q)} \int_0^1 ||P(\alpha(z)) - Q(\beta(z))|| \cdot ||\alpha'(z) + \beta'(z)|| \cdot dz$. If we define $\gamma(t) = (\alpha(t),\beta(t))$, then the first term of the integrand is equal to $h(\gamma(t))$. Next, we reparametrise $\gamma(t)$ in terms of its $\mathcal L_1$ arc length in $R$. For our reparametrised $\gamma$, we get $||\alpha'(z) + \beta'(z)|| = 1$, so the second term of the integrand is equal to 1. Finally, we prove that the parameter $z$ ranges from $0$ to $p+q$, and note that $\gamma(0) = (0,0)$, $\gamma(p+q) = (p,q)$, $\gamma$ is differentiable and non-decreasing, and $||\gamma'(z)||_R = 1$. This gives us the stated formula.
\end{proof}

\subsection{Properties of the parameter space}
\label{subsec:parameter_space_properties}

Before providing the algorithm for minimising our line integral, we first provide some structural insights of our parameter space $R = [0,p] \times [0,q]$. Recall that $P:[0,p] \to \mathbb R$ maps points on the $x$-axis of $R$ to points on the one-dimensional curve $P$, and analogously for $Q$ and the $y$-axis. Hence, each point $(x,y) \in R$ is associated with a pair of points $P(x)$ and $Q(y)$, so that the height function $h(x,y) = ||P(x)-Q(y)||$ is simply the distance between the associated pair of points. Divide the $x$-axis of $R$ into $n-1$ segments that are associated with the $n-1$ segments $p_1p_2,\ldots,p_{n-1}p_n$ of $P$. Divide the $y$-axis of $R$ into $m-1$ segments analogously. In this way, we divide $R$ into $(n-1) (m-1)$ cells, which we label as follows:

\begin{definition}[cell]
	Cell $(i,j)$ is the region of the parameter space associated with segment $p_i p_{i+1}$ on the $x$-axis, and $q_j q_{j+1}$ on the $y$-axis, where $i \in [n-1]$ and $j \in [m-1]$.
\end{definition}

For points $(x,y)$ restricted to a single cell $(i,j)$, the functions $P(x)$ and $Q(y)$ are linear. Hence, $P(x)-Q(y)$ is also linear, so $h(x,y) = ||P(x)-Q(y)||$ is a piecewise linear surface with at most two pieces. If $h(x,y)$ consists of two linear surface pieces, the boundary of these two pieces is along a segment where $h(x,y) = 0$. Since we are working with one-dimensional curves, we have two cases for the relative directions of the vectors $\overrightarrow{p_ip_{i+1}}$ and $\overrightarrow{q_jq_{j+1}}$. If the vectors are in the same direction, since $\overrightarrow{p_ip_{i+1}}$ and $\overrightarrow{q_jq_{j+1}}$ are both parametrised by their arc lengths, they must be travelling in the same direction and at the same rate. Therefore, the line satisfying $h(x,y) = 0$ has slope~$1$ in~$R$. Using a similar argument, if $\overrightarrow{p_ip_{i+1}}$ and $\overrightarrow{q_jq_{j+1}}$ are in opposite direction, then the line satisfying $h(x,y)=0$ has slope~$-1$ in~$R$.

The line with zero height and slope~$1$ will play an important role in our algorithm. We call these lines valleys. If a path $\gamma$ travels along a valley, the line integral accumulates zero cost as long as it remains on the valley, since the valley has zero height.

\begin{definition}[valley]
In a cell, the set of points $(x,y)$ satisfying $h(x,y)=0$ forms a line, moreover, the line has slope~$1$ or~$-1$. We call this line a valley.
\end{definition}

\section{Algorithm} 
\label{sec:algorithm}


Our approach is a dynamic programming algorithm over the cells in the parameter space, which we defined in Section~\ref{subsec:parameter_space}. To the best of our knowledge, all the existing approximation algorithms and heuristics~\cite{DBLP:conf/vldb/BrakatsoulasPSW05,DBLP:conf/gis/BrankovicBKNPW20,koenklaren} use a dynamic programming approach, or simply reduce the problem to a shortest path computation~\cite{DBLP:journals/comgeo/MaheshwariSS18}. Next, we highlight the key difference between our approach and previous approaches.

In previous algorithms, sampling is used along cell boundaries to obtain a discrete set of grid points. Then, the optimal path between the discrete set of grid points is computed. The shortcoming of previous approaches is that an inherent error is introduced by the grid points, where the error depends on the fineness of the grid that is used.

In our algorithm, we consider all points along cell boundaries, not just a discrete subset. However, this introduces a challenge whereby we need to compute optimal paths between continuous segments of points. To overcome this obstacle, we devise a new method of propagating continuous functions across a cell. The main difficulty in analysing the running time of our algorithm lies in bounding the total complexity of the propagated continuous functions, across all cells in the dynamic program.

Our improvement over previous approaches is in many ways similar to previous algorithms for the weighted region problem~\cite{DBLP:journals/jacm/MitchellP91}, and the partial curve matching problem~\cite{DBLP:conf/soda/BuchinBW09}. In all three problems, a line integral is minimised over a given terrain, and an optimal path is computed instead of relying on a sampled grid. However, our problem differs from that of~\cite{DBLP:journals/jacm/MitchellP91} and~\cite{DBLP:conf/soda/BuchinBW09} in two important ways. First, in both~\cite{DBLP:journals/jacm/MitchellP91} and~\cite{DBLP:conf/soda/BuchinBW09}, the terrain is a piecewise constant function, whereas in our problem, the terrain is a piecewise linear function. Second, our main difficulty is bounding the complexity of the propagated functions. In~\cite{DBLP:journals/jacm/MitchellP91}, a different technique is used that does not propagate functions. In~\cite{DBLP:conf/soda/BuchinBW09}, the propagated functions are concave, piecewise linear and their complexities remain relatively low. In our algorithm, the propagated functions are piecewise quadratic and their complexities increase at a much higher, albeit bounded, rate. 

The remainder of our paper is structured as follows. In Section~\ref{subsec:algorithm_setup}, we set up our dynamic program. In Section~\ref{subsec:algorithm_base_case}, we solve its base case. In Section~\ref{subsec:algorithm_propagation_step} we solve the propagation step. In Section~\ref{subsec:running_time_analysis}, we analyse the algorithm's running time. In Section~\ref{subsec:algorithm_bounding_cost_function_complexity}, we fill in missing bounds from the running time analysis. We consider our main technical contribution to be the running time analysis in Sections~\ref{subsec:running_time_analysis} and~\ref{subsec:algorithm_bounding_cost_function_complexity}, and their proofs in Section~\ref{apx:proof_of_claims}.

\subsection{Dynamic program}
\label{subsec:algorithm_setup}

Our dynamic program is performed with respect to the following cost function.

\begin{definition}[cost function]
\label{definition:cost}
Let $(x,y) \in R$, we define
\[
cost(x,y) = \inf_{\gamma \in \Psi(x,y)} \int_0^{x+y} h(\gamma(z)) \cdot dz.
\]
\end{definition}

Recall from Lemma~\ref{lem:cdtw_formula} that $d_{CDTW}(P,Q) = \inf_{\gamma \in \Psi(p,q)} \int_0^{p+q} h(\gamma(z)) \cdot dz$, which implies that $cost(p,q) = d_{CDTW}(P,Q)$. Another way of interpreting Definition~\ref{definition:cost} is that $cost(x,y)$ is equal to $d_{CDTW}(P_x,Q_y)$, where $P_x$ is the subcurve from $P(0)$ to $P(x)$, and $Q_y$ is the subcurve from $Q(0)$ to $Q(y)$.

Recall from Section~\ref{subsec:parameter_space_properties} that the parameter space is divided into $(n-1)(m-1)$ cells. Our dynamic program solves cells one at a time, starting from the bottom left cell and working towards the top right cell. A cell is considered solved if we have computed the cost of every point on the boundary of the cell. Once we solve the top right cell of $R$, we obtain the cost of the top right corner of $R$, which is $cost(p,q) = d_{CDTW}(P,Q)$, and we are done. 

In the base case, we compute the cost of all points lying on the lines $x=0$ and $y=0$. Note that if $x=0$ or $y=0$, then the function $cost(x,y)$ is simply a function in terms of $y$ or $x$ respectively. In general, the function along any cell boundary --- top, bottom, left or right --- is a univariate function in terms of either $x$ or $y$. We call these boundary cost functions.

\begin{definition}[boundary cost function]
A boundary cost function is $cost(x,y)$, but restricted to a top, bottom, left or right boundary of a cell. If it is restricted to a top or bottom (resp. left or right) boundary, the boundary cost function is univariate in terms of $x$ (resp.~$y$).
\end{definition}

In the propagation step, we use induction to solve the cell~$(i,j)$ for all $1 \leq i \leq n-1$ and $1 \leq j \leq m-1$. We assume the base case. We also assume as an inductive hypothesis that, if $i \geq 2$, then the cell~$(i-1,j)$ is already solved, and if $j \geq 2$, then the cell~$(i,j-1)$ is already solved. Our assumptions ensure that we receive as input the boundary cost function along the bottom and left boundaries of the cell~$(i,j)$. In other words, we use the boundary cost functions along the input boundaries to compute the boundary cost functions along the output boundaries. 

\begin{definition}[input/output boundary]
The input boundaries of a cell are its bottom and left boundaries. The output boundaries of a cell are its top and right boundaries.
\end{definition}

We provide details of the base case in Section~\ref{subsec:algorithm_base_case}, and the propagation step in Section~\ref{subsec:algorithm_propagation_step}.

\subsection{Base case}
\label{subsec:algorithm_base_case}

The base case is to compute the cost of all points along the $x$-axis. The $y$-axis can be handled analogously. Recall that $cost(x,y) = \inf_{\gamma \in \Psi(x,y)} \int_0^{x+y} h(\gamma(z)) \cdot dz$. Therefore, for points $(x,0)$ on the $x$-axis, we have $cost(x,0) = \inf_{\gamma \in \Psi(x,0)} \int_0^x h(\gamma(z)) \cdot dz$. Since $\gamma(z)$ is non-decreasing in $x$- and $y$-coordinates, and $||\gamma'(z)|| = 1$, we must have that $\gamma'(z) = (1,0)$. By integrating from $0$ to $z$, we get $\gamma(z)=(z,0)$, which implies that $cost(x,0) = \int_0^x h(z,0) \cdot dz$. 

Consider, for $1 \leq i \leq n-1$, the bottom boundary of the cell $(i,1)$. The height function $h(z)$ is a piecewise linear function with at most two pieces, so its integral $cost(x,0) = \int_0^x h(z,0) \cdot dz$ is a continuous piecewise quadratic function with at most two pieces. Similarly, since the height function along $x=0$ is a piecewise linear function with at most $2(n-1)$ pieces, the boundary cost function along $x=0$ is a continuous piecewise quadratic function with at most $2(n-1)$ pieces. For boundaries not necessarily on the $x$- or $y$-axis, we claim that the boundary cost function is still a continuous piecewise-quadratic function. 

\begin{restatable}{lem}{claimpiecewisequadratic}
\label{claim:piecewise_quadratic}
The boundary cost function is a continuous piecewise-quadratic function.
\end{restatable}

We defer the proof of Lemma~\ref{claim:piecewise_quadratic} to Section~\ref{apx:proof_of_claims}. Although the boundary cost function has at most two pieces for cell boundaries on the $x$- or $y$-axis, in the general case it may have more than two pieces. As previously stated, the main difficulty in bounding our running time analysis in Section~\ref{subsec:running_time_analysis} is to bound complexities of the boundary cost functions.

\subsection{Propagation step}
\label{subsec:algorithm_propagation_step}

First, we define optimal paths in the parameter space. We use optimal paths to propagate the boundary cost functions across cells in the parameter space. Note that the second part of Definition~\ref{defn:optimal_path} is a technical detail to ensure the uniqueness of optimal paths. Intuitively, the optimal path from $s$ to $t$ is the path minimising the path integral, and if there are multiple such paths, the optimal path is the one with maximum $y$-coordinate.

\begin{definition}[optimal path]
\label{defn:optimal_path}
Given $t = (x_t,y_t) \in R$, its optimal path is a path $\gamma \in \Psi(x_t,y_t)$ minimising the integral $\int_0^{x_t+y_t} h(\gamma(z)) \cdot dz$. If there are multiple such curves that minimise the integral, the optimal path is the one with maximum $y$-coordinate (or formally, the one with maximum integral of its $y$-coordinates).
\end{definition}

Suppose $t$ is on the output boundary of the cell $(i,j)$. Consider the optimal path $\gamma$ that starts at $(0,0)$ and ends at $t$. Let $s$ be the first point where $\gamma$ enters the cell $(i,j)$. We consider the subpath from $s$ to $t$, which is entirely contained in the cell $(i,j)$. In the next lemma, we show that the shape of the subpath from $s$ to $t$ is restricted, in particular, there are only three types of paths that we need to consider. A similar proof can be found in Lemma~4 of Maheswari~\etal~\cite{DBLP:journals/comgeo/MaheshwariSS18}. Nonetheless, due to slight differences, we provide a full proof. Specifically, the differences are that we consider one-dimensional curves, and use the $\mathcal L_1$ norm in parameter space to obtain a significantly stronger statement for type $(A)$ paths.

\begin{restatable}{lem}{optimalpathstheorem}
\label{lem:three_types_of_optimal_paths}
Let $t$ be a point on the output boundary of a cell. Let $s$ be the first point where the optimal path to $t$ enters the cell. There are only three types of paths from $s$ to $t$:
\begin{enumerate}[label=(\Alph*)]
	\item The segments of the cell are in opposite directions. Then all paths between $s$ and $t$ have the same cost.
	\item The segments of the cell are in the same direction and the optimal path travels towards the valley, then along the valley, then away from the valley.
	\item The segments of the cell are in the same direction and the optimal path travels towards the valley, then away from the valley.
\end{enumerate}
For an illustration of these three types of paths, see Figure~\ref{fig:three_types_of_optimal_paths}.
\end{restatable}

\begin{figure}[ht]
    \centering
    \includegraphics{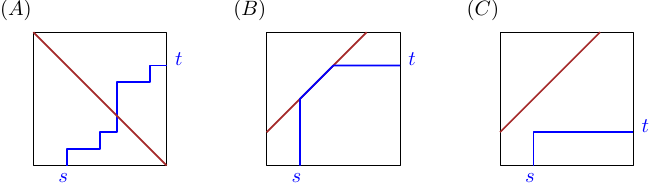}
    \caption{The three types of optimal paths through a cell.}
    \label{fig:three_types_of_optimal_paths}
\end{figure}

\begin{proof}
We begin by summarising the main steps of the proof. Define $\gamma_1$ to be an optimal path to~$s$, followed by either a type $(A)$, $(B)$ or $(C)$ path from $s$ to $t$. If the segments are in opposite directions, we use a type $(A)$ path, whereas if the segments are in the same direction, we use either a type $(B)$ or type $(C)$ path. Define $\gamma_2$ to be an optimal path to~$s$, followed by any path from $s$ to $t$. The main step is to show that $h(\gamma_1(z)) \leq h(\gamma_2(z))$, as this would imply that $\gamma_2$ is an optimal path from $s$ to $t$. In fact, if the segments are in the opposite directions, we get that $h(\gamma_1(z)) = h(\gamma_2(z))$, implying that all type $(A)$ paths from $s$ to $t$ have the same cost. This completes the summary of the main steps of the proof.

We start our proof by considering the case where, for this cell, the curves are in opposite directions. Let $t = (x_t,y_t)$ and let $\gamma_1 \in \Psi(x_t,y_t)$ be an optimal path to $t$. Let $s = (x_s,y_s)$, and let $\gamma_2 \in \Psi(x_t,y_t)$ be the concatenation of an optimal path to $s$ and any path from~$s$ to~$t$. 

By definition, $\gamma_1(x_t+y_t) = \gamma_2(x_t+y_t) = (x_t,y_t)$, and $\gamma_1(x_s+y_s) = \gamma_2(x_s+y_s) = (x_s,y_s)$. It suffices to show that $\int_{x_s+y_s}^{x_t+y_t} h(\gamma_1(z)) \cdot dz = \int_{x_s+y_s}^{x_t+y_t} h(\gamma_2(z)) \cdot dz$, as this would imply that any path from $s$ to $t$ has the same cost. We will show a slightly stronger statement, that for all $x_s+y_s \leq z \leq x_t+y_t$, we have $h(\gamma_1(z)) = h(\gamma_2(z))$. 

Let $\gamma_1(z) = (\alpha_1(z),\beta_1(z))$ be the $x$- and $y$-coordinates of $\gamma_1(z)$. Since $\gamma_1(z) \in \Psi(x_t,y_t)$, we have $||\gamma_1'(z)||_R = 1$ and $\alpha_1'(z),\beta_1'(z) > 0$. Therefore, $\alpha_1'(z) + \beta_1'(z) = 1$. By integrating from $0$ to $z$, we get $\alpha_1(z) + \beta_1(z) = z$. Let $\gamma_2(z) = (\alpha_2(z),\beta_2(z))$. Then similarly, $\alpha_2(z) + \beta_2(z) = z$. 

Without loss of generality, assume $P$ is in the positive direction, and $Q$ is in the negative direction. Recall that $P$ is parametrised by arc length, so the arc length of the subcurve from $P(0)$ to $P(z)$ is $z$. Hence, for $x_s + y_s \leq z \leq x_t + y_t$, and for $i \in \{1,2\}$, since $P$ is in the positive direction, we have $P(\alpha_i(z)) = P(x_s) + \alpha_i(z) - x_s$. Similarly, since $Q$ is in the negative direction, $Q(\beta_i(z)) = Q(y_s) - \beta_i(z) + y_s$. Putting this together,
\[
    \begin{array}{rcl}
            h(\gamma_1(z))
        &=& 
            ||P(\alpha_1(z)) - Q(\beta_1(z))||
        \\
        &=& 
            ||P(x_s) + \alpha_1(z) - x_s - Q(y_s) + \beta_1(z) - y_s||
        \\
        &=& 
            ||P(x_s) + \alpha_2(z) - x_s - Q(y_s) + \beta_2(z) - y_s||
        \\
        &=& 
            ||P(\alpha_2(z)) - Q(\beta_2(z))||
        \\
        &=& 
            h(\gamma_2(z)).
    \end{array}
\]
This concludes the proof in the first case.

Second, we consider the case where, for this cell, the curves are in the same direction. For this case, a similar proof is given in Lemma~4 of Maheshwari~\etal~\cite{DBLP:journals/comgeo/MaheshwariSS18}. Nonetheless, we provide a proof for the sake of completeness. Similarly to the first case, let $t = (x_t,y_t)$ and $s = (x_s,y_s)$. Let $\gamma_1$ be the concatenation of an optimal path to $s$, and then either a type $(B)$ or type $(C)$ path from $s$ to $t$. Recall from the statement of our lemma that a type $(B)$ path is an axis-parallel path from $s$ towards the valley, then a path along the valley, then an axis-parallel path away from the valley to $t$. A type $(C)$ path is an axis-parallel path from $s$ towards the valley, then an axis-parallel path away from the valley to $t$. Let $\gamma_2 \in \Psi(x_t,y_t)$ be the concatenation of an optimal path to $s$, then any path from $s$ to $t$. 

By definition, $\gamma_1(x_t+y_t) = \gamma_2(x_t+y_t) = (x_t,y_t)$, and $\gamma_1(x_s+y_s) = \gamma_2(x_s+y_s) = (x_s,y_s)$. We are required to show that $\int_{x_s+y_s}^{x_t+y_t} h(\gamma_1(z)) \cdot dz \leq \int_{x_s+y_s}^{x_t+y_t} h(\gamma_2(z)) \cdot dz$, as this would imply that $\gamma_1$ is an optimal path from $s$ to $t$. We will show a slightly stronger statement, that for all $x_s + y_s \leq z \leq x_t + y_t$, we have $h(\gamma_1(z)) \leq h(\gamma_2(z))$. 

For $i \in \{1,2\}$, let $\gamma_i(z) = (\alpha_i(z), \beta_i(z))$. For the same reasons as in the first case, $\alpha_i(z) + \beta_i(z) = z$. Without loss of generality, assume both $P$ and $Q$ are in the positive direction. For the same reasons as in the first case, for $i \in \{1,2\}$, we have $P(\alpha_i(z)) = P(x_s) + \alpha_i(z) - x_s$ and $Q(\beta_i(z)) = Q(y_s) + \beta_i(z) - y_s$.

For type $(B)$ paths, if $\gamma_1(z)$ is along the valley, then $h(\gamma_1(z)) = 0$ so clearly $h(\gamma_1(z)) \leq h(\gamma_2(z))$. For both type $(B)$ and type $(C)$ paths, the only remaining case is where $\gamma_1(z)$ is along an axis parallel path from $s$ towards the valley, since other the case where $\gamma_1(z)$ is along an axis parallel path away from the valley to $t$ can be handled analogously.

Without loss of generality, assume $s$ is on the bottom boundary, and $\gamma_1(z)$ is along a vertical path from $s$ to the valley. Then $\alpha_1(z) = x_s$ and $\beta_1(z) = z - \alpha_1(z) = z - x_s$. Since $\gamma_2(z)$ is monotonically increasing in the $x$-coordinate, we have $\alpha_2(z) \geq x_s$ and $\beta_2(z) = z - \alpha_2(z) \leq z - x_s$. Without loss of generality, assume $P(x_s) - Q(y_s) \geq 0$. Since $\gamma_1(z)$ is on the vertical path from $s$ to the valley, but does not cross over the valley, we have $P(\alpha_1(z)) - Q(\beta_1(z)) \geq 0$. Putting this together, 
\[
    \begin{array}{rcl}
            h(\gamma_1(z))
        &=& 
            P(\alpha_1(z)) - Q(\beta_1(z))
        \\
        &=& 
            P(x_s) + \alpha_1(z) - x_s - Q(y_s) - \beta_1(z) + y_s
        \\
        &=& 
            P(x_s) + x_s - x_s - Q(y_s) - z + x_s + y_s
        \\
        &\leq& 
            P(x_s) + \alpha_2(z) - x_s - Q(y_s) - \beta_2(z) - y_s
        \\
        &=& 
            P(\alpha_2(z)) - Q(\beta_2(z))
        \\
        &=& 
            h(\gamma_2(z)),
    \end{array}
\]
where the second last equality uses that $P(\alpha_2(z)) - Q(\beta_2(z)) \geq P(\alpha_1(z)) - Q(\beta_1(z)) \geq 0$. This concludes the proof in the second case, and we are done.
\end{proof}

We leverage Lemma~\ref{lem:three_types_of_optimal_paths} to propagate the boundary cost function from the input boundaries to the output boundaries of a cell. We provide an outline of our propagation procedure in one of the three cases, that is, for type $(B)$ paths. These paths are the most interesting to analyse, and looking at this special case provides us with some intuition for the other cases.
For type $(B)$ paths, we compute the cost function along the output boundary in three consecutive steps. We first list the steps, then we describe the steps in detail.

\begin{enumerate}
	\item We compute the cost function along the valley in a restricted sense.
	\item We compute the cost function along the valley in general.
	\item We compute the cost function along the output boundary.
\end{enumerate}

In the first step, we restrict our attention only to paths that travel from the input boundary towards the valley. This is the first segment in the type $(B)$ path as defined in Lemma~\ref{lem:three_types_of_optimal_paths}. We call this first segment a type $(B1)$ path, see Figure~\ref{fig:s2v_simplified}. Define the type $(B1)$ cost function to be the cost function along the valley if we can only use type $(B1)$ paths from the input boundary to the valley.
The type $(B1)$ cost function is simply the cost function along the bottom or left boundary plus the integral of the height function along the type $(B1)$ path. The height function along the type $(B1)$ path is a linear function, so the integral is a quadratic function. To obtain the type $(B1)$ cost function, we add the quadratic function for the type $(B1)$ path to the cost function along an input boundary. We combine the type $(B1)$ cost functions along the bottom and the left boundaries by taking their lower envelope.

\begin{figure}[ht]
    \centering
    \includegraphics{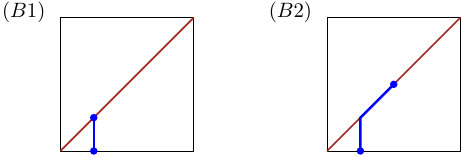}
    \caption{The type $(B1)$ and type $(B2)$ paths from the bottom boundary to the valley.}
    \label{fig:s2v_simplified}
\end{figure}

In the second step, we compute the cost function along the valley in general. It suffices to consider paths that travel from the input boundary towards the valley, and then travel along the valley. This path is the first two segments in a type $(B)$ path as defined in  Lemma~\ref{lem:three_types_of_optimal_paths}.
We call these first two segments a type $(B2)$ path, see Figure~\ref{fig:s2v_simplified}.
Since the height function is zero along the valley, if we can reach a valley point with a particular cost with a type $(B1)$ path, then we can reach all points on the valley above and to the right of it with a type $(B2)$ path with the same cost. Therefore, the type $(B2)$ cost function is the cumulative minimum of the type $(B1)$ cost function, see Figure~\ref{fig:type_b2_cost_function_main_paper}. Note that the type $(B2)$ cost function may have more quadratic pieces than the type $(B1)$ cost function. For example, in Figure~\ref{fig:type_b2_cost_function_main_paper}, the type $(B2)$ cost function has twice as many quadratic pieces as the type $(B1)$ cost function, since each quadratic piece in the type $(B1)$ cost function splits into two quadratic pieces in the type $(B2)$ cost function --- the original quadratic piece plus an additional horizontal piece. 

\begin{figure}[ht]
    \centering
    \includegraphics[width=0.4\textwidth]{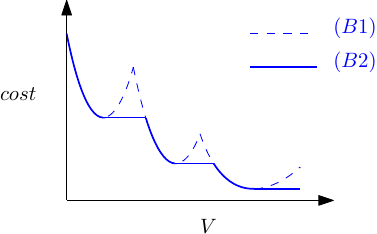}
	\caption{The type $(B2)$ cost function plotted over its position along the valley $V$. The type $(B2)$ cost function is the cumulative minimum of the type $(B1)$ cost function.}
    \label{fig:type_b2_cost_function_main_paper}
\end{figure}

In the third step, we compute the cost function along the output boundary, given the type $(B2)$ cost function along the valley. A type $(B)$ path is a type $(B2)$ path appended with a horizontal or vertical path from the valley to the boundary. The height function of the appended path is a linear function, so its integral is a quadratic function. We add this quadratic function to the type $(B2)$ function along the valley to obtain the output function.
This completes the description of the propagation step in the type $(B)$ paths case.

Using a similar approach, we can compute the cost function along the output boundary in the type $(A)$ and type $(C)$ paths as well. The propagation procedure differs slightly for each of the three path types, for details see Section~\ref{apx:proof_of_claims}. Recall that due to the second step of the type $(B)$ propagation, each quadratic piece along the input boundary may propagate to up to two pieces along the output boundary. In general, we claim that each quadratic piece along the input boundary propagates to at most a constant number of pieces along the output boundary. Moreover, given a single input quadratic piece, this constant number of output quadratic pieces can be computed in constant time. 

\begin{restatable}{lem}{claimpropagate}
\label{claim:propagate}
Each quadratic piece in the input boundary cost function propagates to at most a constant number of pieces along the output boundary. Propagating a quadratic piece takes constant time. 
\end{restatable}

We defer the proof of Lemma~\ref{claim:propagate} to Section~\ref{apx:proof_of_claims}. We can now state our propagation step in general. Divide the input boundaries into segments, so that for each segments, the cost function along that segments is a single quadratic piece. Apply Lemma~\ref{claim:propagate} to a segments to compute in constant time a piecewise quadratic cost function along the output boundary. Apply this process to all segments to obtain a set of piecewise quadratic cost functions along the output boundary. Combine these cost functions by taking their lower envelope. Return this lower envelope as the boundary cost function along the output boundary. 
This completes the statement of our propagation step. Its correctness follows from construction.

\subsection{Running time analysis}
\label{subsec:running_time_analysis}

We start the section with a useful lemma. Essentially the same result is stated without proof as Observation~3.3 in~\cite{DBLP:conf/soda/BuchinBW09}. For the sake of completeness, we provide a proof.

\begin{restatable}{lem}{crossinglemma}
\label{lem:crossing_lemma_main_paper}
Let $\gamma_1, \gamma_2$ be two optimal paths. These paths cannot cross, i.e., there are no $z_1, z_2$ such that $\gamma_1(z_1)$ is below $\gamma_2(z_1)$ and $\gamma_1(z_2)$ is above $\gamma_2(z_2)$.
\end{restatable}

\begin{figure}[ht]
    \centering
    \includegraphics{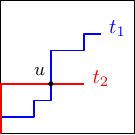}
    \caption{A pair of optimal paths that cross at a point $u$.}
    \label{fig:crossing_lemma_appendix}
\end{figure}

\begin{proof}
Suppose for sake of contradiction that there exists a pair of optimal paths, $\gamma_1$ to $t_1$ and $\gamma_2$ to $t_2$, that cross over at a point $u$, see Figure~\ref{fig:crossing_lemma_appendix}. Moreover, assume that $u$ is the first such crossover point, i.e., $u$ is the last point such that $\gamma_1$ and $\gamma_2$ are non-crossing until $u$. Without loss of generality, before the intersection point $u$, the path $\gamma_1$ is below $\gamma_2$, and after the intersection point $u$, the path $\gamma_1$ is above $\gamma_2$. Since $\gamma_1$ is an optimal path to $t_1$, the portion of $\gamma_1$ up to $u$ must be the optimal path to $u$. Similarly, the portion of $\gamma_2$ up to $u$ must be the optimal path to $u$. Since optimal paths are unique, we obtain that $\gamma_1$ and $\gamma_2$ are identical up to the first crossover point $u$, contradicting the fact that $\gamma_1$ and $\gamma_2$ cross.
\end{proof}

Define $N$ to be the total number of quadratic pieces in the boundary cost functions over all boundaries of all cells. We will show that the running time of our algorithm is $\Oh(N)$.

\begin{lemma}
\label{lem:running_time_analysis_Tn}
The running time of our dynamic programming algorithm is $\Oh(N)$.
\end{lemma}

\begin{proof}
The running time of the dynamic program is dominated by the propagation step. Let $I_{i,j}$ denote the input boundaries of the cell $(i,j)$. Let $|I_{i,j}|$ denote the number of quadratic functions in the input boundary cost function. By Lemma~\ref{claim:propagate}, each piece only propagates to a constant number of new pieces along the output boundary, and these pieces can be computed in constant time. The final piecewise quadratic function is the lower envelope of all the new pieces, of which there are $\Oh(|I_{i,j}|)$ many. 

We use Lemma~\ref{lem:crossing_lemma_main_paper} to speed up the computation of the lower envelope, so that this step takes only $\Oh(|I_{i,j}|)$ time. Since optimal paths do not cross, it implies that the new pieces along the output boundary appear in the same order as their input pieces. We perform the propagation in order of the input pieces. We maintain the lower envelope of the new pieces in a stack. For each newly propagated piece, we remove the suffix that is dominated by the new piece and then add the new piece to the stack. Since each quadratic piece can be added to the stack at most once, and removed from the stack at most once, the entire operation takes  $\Oh(|I_{i,j}|)$ time. Summing over all cells, we obtain an overall running time of $\Oh(N)$.
\end{proof}

Note that Lemma~\ref{lem:running_time_analysis_Tn} does not yet guarantee that our algorithm runs in polynomial time as we additionally need to bound $N$ by a polynomial.
Lemma~\ref{claim:propagate} is of limited help. The lemma states that each piece on the input boundary propagates to at most a constant number of pieces on the output boundary. Recall that in Section~\ref{subsec:algorithm_propagation_step}, we illustrated a type $(B)$ path that resulted in an output boundary having twice as many quadratic pieces as its input boundary. The doubling occurred in the second step of the propagation of type $(B)$ paths, see Figure~\ref{fig:type_b2_cost_function_main_paper}. If this doubling behaviour were to occur for all our cells in our dynamic program, we would get up to $N = \Omega(2^{n+m})$ quadratic pieces in the worst case, where $n$ and $m$ are the complexities of the polygonal curves $P$ and $Q$.
To obtain a polynomial running time, we show that although this doubling behaviour may occur, it does not occur \emph{too often}. 

\subsection{Bounding the cost function's complexity}

\label{subsec:algorithm_bounding_cost_function_complexity}

Our bound comes in two parts. First, we subdivide the boundaries in the parameter space into subsegments and show, in Lemma \ref{lem:number_of_critical_points_on_Ak_main_paper}, that there are $\Oh((n+m)^3)$ subsegments in total. Second, in Lemma \ref{lem:number_of_critical_points_on_Akl_main_paper}, we show that each subsegment has at most $\Oh((n+m)^2)$ quadratic pieces. Putting this together in Theorem \ref{thm:main_main_paper} gives $N = \Oh((n+m)^5)$. 

We first define the $\Oh((n+m)^3)$ subsegments. The intuition behind the subsegments is that for any two points on the subsegment, the optimal path to either of those two points is structurally similar. We can deform one of the optimal paths to the other without passing through any cell corner, or any points where a valley meets a boundary.

Formally, define $A_k$ to be the union of the input boundaries of the cells $(i,j)$ such that $i+j=k$. Alternatively, $A_k$ is the union of the output boundaries of the cells $(i,j)$ such that $i+j+1=k$. Next, construct the partition $A_k \coloneqq \{ A_{k,1}, A_{k, 2}, \dots, A_{k, L} \}$ of $A_k$ into subsegments. Define the subsegment $A_{k,\ell}$ to be the segment between the $\ell^{th}$ and $(\ell+1)^{th}$ \emph{critical point} along $A_k$. We define a critical point to be either \emph{(i)} a cell corner, \emph{(ii)} a point where the valley meets the boundary, or \emph{(iii)} a point where the optimal path switches from passing through a subsegment $A_{k-1,\ell'}$ to a different subsegment $A_{k-1,\ell''}$. 

Let $|A_k|$ denote the number of piecewise quadratic cost functions $A_{k,\ell}$ along $A_k$. Let $|A_{k,\ell}|$ denote the number of pieces in the piecewise quadratic cost function along the subsegment~$A_{k,\ell}$. Thus, we can rewrite the total number of quadratic functions $N$ as:
\[
	N = \sum_{k=2}^{n+m-1} \sum_{\ell=1}^{|A_k|} |A_{k,\ell}|.
\]

We first show that the number of subsegments $|A_k|$ is bounded by $\Oh(k^2)$ and then proceed to show that $|A_{k,\ell}|$ is bounded by $\Oh(k^2)$ for all $k, \ell$.
\begin{lemma}
\label{lem:number_of_critical_points_on_Ak_main_paper}
For any $k \in [n + m]$, we have $|A_k| \leq 2k^2$.
\end{lemma}

\begin{proof}
	We prove the lemma by induction. Since the cell $(1,1)$ has at most one valley, and since the input boundary $A_2$ has one cell corner, we have $|A_2| \leq 3$. For the inductive step, note that there are at most $2k$ cell corners on $A_k$, and there are at most $k$ points where a valley meets a boundary on $A_k$. By the inductive hypothesis, there are at most $2 (k-1)^2$ subsegments on $A_{k-1}$. And as optimal paths do not cross by Lemma~\ref{lem:crossing_lemma_main_paper}, each subsegment of $A_{k-1}$ contributes at most once to the optimal path switching from one subsegment to a different one on $A_k$. Thus, for $k \geq 3$, we obtain $|A_k| \leq 2(k-1)^2 + 2k + k + 1 = 2k^2 - k + 3 \leq 2k^2.$
\end{proof}

Next, we show that $|A_{k,\ell}|$ is bounded by $\Oh(k^2)$ for all $k, \ell$. We proceed by induction. Recall that, due to the third type of critical point, all optimal paths to $A_{k,\ell}$ pass through the same subsegment of $A_{k-1}$, namely $A_{k-1,\ell'}$ for some $\ell'$. Our approach is to assume the inductive hypothesis for $|A_{k-1,\ell'}|$, and bound $|A_{k,\ell}|$ relative to $|A_{k-1,\ell'}|$. We already have a bound of this form, specifically,  Lemma~\ref{claim:propagate} implies that $|A_{k,\ell}| \leq c \cdot |A_{k-1,\ell'}|$, for some constant $c > 1$. Unfortunately, this bound does not rule out an exponential growth in the cost function complexity. We instead prove the following improved bound:

\begin{restatable}{lem}{claimdistinctpairs}
\label{claim:distinct_pairs}
Let $|A_{k,\ell}|$ be a subsegment on $A_k$, and suppose all optimal paths to $|A_{k,\ell}|$ pass through subsegment $|A_{k-1,\ell'}|$ on $A_{k-1}$. Then
\[
    \begin{array}{rcl}
    |A_{k,\ell}| &\leq& |A_{k-1,\ell'}| + D(A_{k-1,\ell'})+1, \\
    D(A_{k,\ell}) &\leq& D(A_{k-1,\ell'}) + 1,
    \end{array}
\]
where $D(\cdot)$ counts, for a given subsegment, the number of distinct pairs $(a,b)$ over all quadratics $ax^2 + bx + c$ in the boundary cost function for that subsegment.
\end{restatable} 

We defer the proof of Lemma~\ref{claim:distinct_pairs} to Section~\ref{apx:proof_of_claims}. The lemma obtains a polynomial bound on the growth of the number of quadratic pieces by showing, along the way, a polynomial bound on the growth of the number of distinct $(a,b)$ pairs over the quadratics $ax^2 + bx + c$. 

As we consider this lemma to be one of the main technical contributions of the paper, we will briefly outline its intuition. It is helpful for us to revisit the doubling behaviour of type $(B)$ paths. Recall that in our example in Figure~\ref{fig:type_b2_cost_function_main_paper}, we may have  $|A_{k,\ell}| = 2|A_{k-1,\ell'}|$. This doubling behaviour does not contradict Lemma~\ref{claim:distinct_pairs}, so long as all quadratic functions along $A_{k-1,\ell'}$ have distinct $(a,b)$ pairs. In fact, for $|A_{k,\ell}| = 2|A_{k-1,\ell'}|$ to occur, each quadratic function in $|A_{k-1,\ell'}|$ must create a new horizontal piece in the cumulative minimum step. But for any two quadratic functions with the same $(a,b)$ pair, only one of them can to create a new horizontal piece, since the horizontal piece starts at the $x$-coordinate~$-\frac {b} {2a}$. Therefore, we must have had that all quadratic functions along $A_{k-1,\ell'}$ have distinct $(a,b)$ pairs. In Section~\ref{apx:proof_of_claims}, we generalise this argument and prove $|A_{k,\ell}| \leq |A_{k-1,\ell'}| + D(A_{k-1,\ell'})+1$.

We perform a similar analysis in the special case of type $(B)$ paths to give the intuition behind $D(A_{k,\ell}) \leq D(A_{k-1,\ell'}) + 1$. For type $(B)$ paths, the number of distinct $(a,b)$ pairs changes only in the cumulative minimum step. All pieces along $A_{k,\ell}$ can either be mapped to a piece along $A_{k-1,\ell'}$, or it is a new horizontal piece. However, all new horizontal pieces have an $(a,b)$ pair of $(0,0)$, so the number of distinct $(a,b)$ pairs increases by only one. For the full proof of Lemma~\ref{claim:distinct_pairs} for all three path types, refer to Section~\ref{apx:proof_of_claims}.

With Lemma~\ref{claim:distinct_pairs} in mind, we can now prove a bound on $|A_{k,\ell}|$ by induction.

\begin{lemma} \label{lem:number_of_critical_points_on_Akl_main_paper}
For any $k \in [n+m]$ and $A_{k,\ell} \in A_k$ we have $|A_{k,\ell}| \leq k^2$.
\end{lemma}
\begin{proof}
Note that in the base case $D(A_{2,\ell}) \leq 2$ and $|A_{2,\ell}| \leq 4$ for any $A_{2,\ell} \in A_2$. By Lemma~\ref{claim:distinct_pairs}, we get $D(A_{k,\ell}) \leq D(A_{k-1,\ell'}) + 1$, for some subsegment $A_{k-1,\ell'}$ on $A_{k-1}$. By a simple induction, we get $D(A_{k,\ell}) \leq k$ for any $k \in [n + m]$. Similarly, assuming $|A_{k-1,\ell}| \leq (k-1)^2$ for any $A_{k-1,\ell} \in A_{k-1}$, we use Lemma~\ref{claim:distinct_pairs} to inductively obtain
$|A_{k,\ell}| \leq |A_{k-1,\ell'}| + D(A_{k-1,\ell'}) + 1 \leq |A_{k-1,\ell'}|  + k \leq (k-1)^2 + k-1 + 1 \leq k^2$ for any $A_{k,\ell} \in A_k$.
\end{proof}

Using our lemmas, we can finally bound $N$, and thereby the overall running time.

\begin{theorem} \label{thm:main_main_paper}
The Continuous Dynamic Time Warping distance between two 1-dimensional polygonal curves of length $n$ and $m$, respectively, can be computed in time $\Oh((n+m)^5)$.
\end{theorem}

\begin{proof}
Using Lemmas~\ref{lem:number_of_critical_points_on_Ak_main_paper}~and~\ref{lem:number_of_critical_points_on_Akl_main_paper}, we have
\begin{align*}
	N = \sum_{k=2}^{n+m-1} \sum_{\ell=1}^{|A_k|} |A_{k,\ell}|
	\leq \sum_{k=2}^{n+m} \sum_{\ell=1}^{|A_k|} k^2
	\leq \sum_{k=2}^{n+m} 2k^4 
	\leq 2(n+m)^5.
\end{align*}
Thus, the overall running time of our algorithm is $\Oh((n+m)^5)$, by Lemma~\ref{lem:running_time_analysis_Tn}.
\end{proof}

\section{Proofs of Lemmas~\ref{claim:piecewise_quadratic}, \ref{claim:propagate} and~\ref{claim:distinct_pairs}}

\label{apx:proof_of_claims}

In this section, we provide proofs for the following three lemmas.

\claimpiecewisequadratic*
\claimpropagate*
\claimdistinctpairs*

Our approach is to use induction to prove these three lemmas. Our inductive approach is to prove the three lemmas simultaneously, in that the inductive hypothesis for one lemma is required to prove the other lemmas. In fact, for this inductive approach to work, we require a stronger version of Lemma~\ref{claim:piecewise_quadratic}, which we state as Lemma~\ref{claim:details_piecewise_quadratic}.

\begin{lemma}
\label{claim:details_piecewise_quadratic}
The cost along a subsegment $A_{k,\ell}$ is a continuous piecewise quadratic function. Moreover, for every point $t$ on the piecewise quadratic function, the left derivative at $t$ is greater than or equal to the right derivative at $t$.
\end{lemma}

The remainder of this paper is dedicated to prove these lemmas required for the correctness of our algorithm and its running time. In Section~\ref{subsec:algorithm_base_case}, we proved Lemma~\ref{claim:details_piecewise_quadratic} for the base case where $k=2$. Lemmas~\ref{claim:propagate} and~\ref{claim:distinct_pairs} are clearly true for $k=2$. For the inductive case we require a case analysis.  Recall Lemma~\ref{lem:three_types_of_optimal_paths}.

\optimalpathstheorem*

Using Lemma~\ref{lem:three_types_of_optimal_paths}, we divide our inductive step into separate case analyses for each of the three types of optimal paths:

\begin{figure}[ht]
    \centering
    \begin{tabular}{|l|l|l|l|}
        \hline
         & Type $(A)$ & Type $(B)$ & Type $(C)$ \\
        \hline
        Lemma~\ref{claim:details_piecewise_quadratic}
        & Lemma~\ref{lem:type_A_piecewise_quadratic}
        & Lemma~\ref{lem:type_B_piecewise_quadratic}
        & Lemma~\ref{claim:type_C_piecewise_quadratic}
        \\
        \hline
        Lemma~\ref{claim:propagate}
        & Lemma~\ref{lem:type_A_propagate_in_constant_time}
        & Lemma~\ref{lem:type_B_propagate_in_constant_time}
        & Lemma~\ref{claim:type_C_propagate_in_constant_time}
        \\
        \hline
        Lemma~\ref{claim:distinct_pairs}
        & Lemma~\ref{lem:type_A_bounds}
        & Lemma~\ref{lem:type_B_bounds}
        & Lemma~\ref{claim:type_C_bounds}
        \\
        \hline
    \end{tabular}
\end{figure}

For proving our lemmas, it will be useful to us to introduce the following notation.

\begin{definition}[parent]
\label{defn:parent}
If all optimal paths of $A_{k,\ell}$ pass through $A_{k-1,\ell'}$, then the subsegment $A_{k-1,\ell'}$ is called the parent of $A_{k,\ell}$.
\end{definition}

\subsection{Type \textit{(A)} paths}
\label{sec:case_analysis_type_a_paths}

In this section we show that Lemmas \ref{claim:propagate}, \ref{claim:distinct_pairs}, and~\ref{claim:details_piecewise_quadratic} hold for the output boundary cost functions of type $(A)$ paths. We assume the inductive hypothesis, that Lemmas \ref{claim:propagate}, \ref{claim:distinct_pairs}, and~\ref{claim:details_piecewise_quadratic} hold for the input boundary cost function. First we state two observations that help simplify our proofs.

\begin{observation}
\label{obs:type_A_all_paths}
Let $A_{k-1,\ell'}$ be the parent of $A_{k,\ell}$.
Suppose there exists a type $(A)$ path from $A_{k-1,\ell'}$ to $A_{k,\ell}$. Then all paths from $A_{k-1,\ell'}$ to $A_{k,\ell}$ are type $(A)$ paths.
\end{observation}

\begin{proof}
Type $(A)$ paths only exist between segments that are in opposite directions. Type $(B)$ or $(C)$ paths only exist between segments that are in the same direction. If not all paths were type $(A)$ paths, we must have switched directions, which means there is a cell corner, which is a critical point. But in a subsegment there are no critical points.
\end{proof}

\begin{observation}
\label{obs:type_A_hor_or_vert}
All optimal paths in a type $(A)$ cell can be replaced by a single horizontal segment or a single vertical segment, possibly changing the starting point, but without changing the cost at the end point.
\end{observation}

\begin{proof}
Consider any optimal path through a type $(A)$ cell and without loss of generality assume that it starts on the bottom boundary. By Lemma~\ref{lem:three_types_of_optimal_paths} we can replace a type (A) path by any other path without changing the cost. In particular, we can replace it by the path that first goes horizontally along the bottom boundary and then vertically until the end point. However, the cost of this path is at least as high as the exclusively vertical path that starts at the point where we leave the bottom boundary. Thus, we can only consider paths that consist of a single horizontal or vertical segment, without changing the cost on the outputs of the cell.
\end{proof}

We now show Lemma~\ref{claim:details_piecewise_quadratic} for type $(A)$ paths.
Note that by Observation \ref{obs:type_A_all_paths} we already know that either all paths are of type $(A)$ or none are. Hence, we only have to consider the former case.

\begin{lemma}
\label{lem:type_A_piecewise_quadratic}
Let $A_{k-1,\ell'}$ be the parent of $A_{k,\ell}$. Suppose all paths from $A_{k-1,\ell'}$ to $A_{k,\ell}$ are type $(A)$ paths. Then the cost along $A_{k,\ell}$ is a continuous piecewise quadratic function. Moreover, for every point $t$ on the piecewise quadratic function, the left derivative at $t$ is greater than or equal to the right derivative at $t$.
\end{lemma}

\begin{proof}
The cost along $A_{k,\ell}$ is equal to the cost along $A_{k-1,\ell'}$ plus the cost of an optimal type $(A)$ path.
By Observation \ref{obs:type_A_hor_or_vert}, all optimal paths that we have to consider are either a single vertical or horizontal segment.
The cost of a vertical/horizontal path inside a subsegment is a quadratic with respect to its $x$-coordinate/$y$-coordinate. This quadratic cost may change at a critical point, but since there are no critical points inside a subsegment, the cost is simply a quadratic. Summing a function with a quadratic preserves the fact that the function is a piecewise quadratic function. Summing a function with a quadratic also preserves the property that for every point $t$ on the piecewise quadratic function, the left derivative at $t$ is greater than or equal to the right derivative at $t$.
\end{proof}

We now show Lemma~\ref{claim:propagate} for type $(A)$ paths. Again, by Observation \ref{obs:type_A_all_paths}, we can assume that all paths are of type $(A)$.

\begin{lemma}
\label{lem:type_A_propagate_in_constant_time}
Let $A_{k-1,\ell'}$ be the parent of $A_{k,\ell}$. Suppose all paths starting from $A_{k-1,\ell'}$ are type $(A)$ paths. If we propagate a single quadratic piece of $A_{k-1,\ell'}$ to the next level $A_k$, it is a piecewise quadratic function with at most a constant number of pieces. Moreover, this propagation step takes only constant time.
\end{lemma}

\begin{proof}
First we consider propagating $A_{k-1,\ell'}$ to the opposite side. By Observation \ref{obs:type_A_hor_or_vert}, we can assume that this happens along a vertical or horizontal segment. These paths have piecewise quadratic cost with at most two pieces. Computing the cost from $A_{k-1,\ell'}$ to the opposite side takes constant time. In the case that $A_{k-1,\ell'}$ contains the top left corner of its cell, we need to propagate its cost along the top boundary. This cost has at most two pieces and takes constant time to compute. In the case that $A_{k-1,\ell'}$ contains the bottom right corner of its cell, we need to propagate its cost along the right boundary. This cost has at most two pieces and takes constant time to compute. Hence, the cost along the next level $A_k$ has a constant number of pieces and the propagation can be computed in constant time.
\end{proof}

Finally, we show Lemma~\ref{claim:distinct_pairs} for type $(A)$ paths.

\begin{lemma}
\label{lem:type_A_bounds}
Let $A_{k-1,\ell'}$ be the parent of $A_{k,\ell}$. Suppose all paths from $A_{k-1,\ell'}$ to $A_{k,\ell}$ are type $(A)$ paths, then $$D(A_{k,\ell}) \leq D(A_{k-1,\ell'}),$$ $$|A_{k,\ell}| \leq |A_{k-1,\ell'}|.$$
\end{lemma}

\begin{proof}
	Without loss of generality $A_{k-1,\ell'}$ is on the bottom boundary. If $A_{k,\ell}$ is on the top boundary, then its cost is the same as the one on $A_{k-1,\ell'}$ plus a quadratic, and we have $D(A_{k,\ell}) \leq D(A_{k-1,\ell'})$ and $|A_{k,\ell}| \leq |A_{k-1,\ell'}|$. If $A_{k,\ell}$ is on the right boundary, then all vertical paths pass through the bottom right corner and thus there is a single path with minimum cost. This implies $D(A_{k,\ell}) = |A_{k,\ell}| = 1$.
\end{proof}

\subsection{Separating type \textit{(B)} and \textit{(C)} paths}

Similarly to Section~\ref{sec:case_analysis_type_a_paths}, it will be useful to show Lemmas \ref{claim:propagate}, \ref{claim:distinct_pairs}, and~\ref{claim:details_piecewise_quadratic}, for output boundary cost functions of only type $(B)$ or type $(C)$ paths. It will be useful to be able to consider two types of paths separately. In order to do this, we must show that we can further divide our subsegments so that, for each subsegment on the output boundary, all optimal paths to the subsegments are type $(B)$ paths or that all optimal paths are type $(C)$ paths.

\begin{lemma}
\label{lem:add_critical_points_between_type_b_type_c}
Suppose we add a critical point in $A_{k,\ell}$ where the optimal path from $A_{k-1,\ell'}$ to $A_{k,\ell}$ changes between a type $(B)$ path and a type $(C)$ path. Then at most two critical points are added to $A_{k,\ell}$. 
\end{lemma}

\begin{proof}
Suppose there were at least three such critical points. Then the path types must alternate, between $B$, $C$, $B$, $C$. In particular, there must exist, along the output boundary, paths of type $B,C,B$. The two outer type $(B)$ paths must visit the valley. However, the inner type $(C)$ path is sandwiched between the type $(B)$ paths, since by Lemma~\ref{lem:crossing_lemma_main_paper}, no optimal paths may cross. Therefore, the type $(C)$ path must also visit the valley, which contradicts the definition of the type $(C)$ path.
\end{proof}

\subsection{Type \textit{(B)} paths}
\label{sec:b}

As a consequence of the additional critical points in Lemma~\ref{lem:add_critical_points_between_type_b_type_c}, we obtain the following fact. 

\begin{observation}
\label{obs:type_B_all_paths}
Let $A_{k-1,\ell'}$ be the parent of $A_{k,\ell}$. Suppose there exists a type $(B)$ path from $A_{k-1,\ell'}$ to $A_{k,\ell}$. Then all paths from $A_{k-1,\ell'}$ to $A_{k,\ell}$ are type $(B)$ paths.
\end{observation}

Recall that $V$ is the valley of our cell in the parameter space, if one exists.

\begin{definition}
Let the type $(B1)$ function be the cost function along $V$ if we only allow an axis parallel path from $A_{k-1,\ell'}$ to $V$. Let the type $(B2)$ function be the cost function along $V$ if we allow an axis parallel path from $A_{k-1,\ell'}$ to $V$ followed by a path of slope~$1$ along~$V$.
\end{definition}

\begin{figure}[ht]
    \centering
    \includegraphics{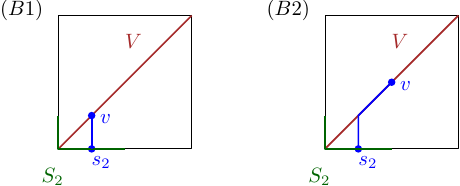}
    \caption{The type $(B1)$ and type $(B2)$ optimal path from $s_2 \in S_2$ to $v \in V$.}
    \label{fig:s2v}
\end{figure}

We prove an observation we made in main paper, that the type $(B2)$ cost function is indeed the cumulative minimum of the type $(B1)$ cost function.

\begin{observation}
The type $(B2)$ cost function is the cumulative minimum function of the type $(B1)$ cost function.
\end{observation}

\begin{proof}
We observe that a type $(B2)$ paths is a type $(B1)$ path appended with a path along the valley. Since the height function is zero along the valley, if we can reach a valley point for a particular cost with a type $(B1)$ path, then we can reach all points on the valley above and to the right of it with a type $(B2)$ path. So the type $(B2)$ cost function is the cumulative minimum of the type $(B1)$ cost function, see Figure~\ref{fig:type_b2_cost_function_main_paper}. \end{proof}

\begin{figure}[ht]
    \centering
    \includegraphics[width=0.4\textwidth]{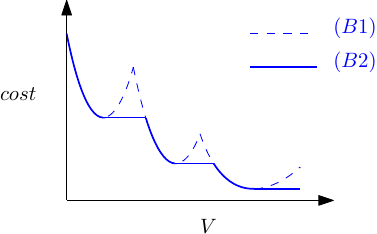}
    \caption{The type $(B2)$ cost function.}
    \label{fig:type_b2_cost_function}
\end{figure}

Finally, we can show the main lemmas of this section. Lemmas~\ref{lem:type_B_piecewise_quadratic}, \ref{lem:type_B_propagate_in_constant_time} and~\ref{lem:type_B_bounds} are the special cases of Lemmas~\ref{claim:details_piecewise_quadratic}, \ref{claim:propagate} and~\ref{claim:distinct_pairs} in the type $(B)$ paths case.

\begin{lemma}
\label{lem:type_B_piecewise_quadratic}
Let $A_{k-1,\ell'}$ be the parent of $A_{k,\ell}$. Suppose all paths from $A_{k-1,\ell'}$ to $A_{k,\ell}$ are type $(B)$ paths. Then the cost along $A_{k,\ell}$ is a continuous piecewise quadratic function. Moreover, for every point $t$ on the piecewise quadratic function, the left derivative at $t$ is greater than or equal to the right derivative at $t$.

\end{lemma}

\begin{proof}
Three operations are performed to get from the cost function along $A_{k-1,\ell'}$ to the cost function along $A_{k,\ell}$. From $A_{k-1,\ell'}$ to the type $(B1)$ cost function we add a quadratic function. From type $(B1)$ to type $(B2)$, we take the cumulative minimum, as shown in Figure~\ref{fig:s2v}. From the type $(B2)$ cost function to $A_{k,\ell}$, we add a quadratic function. All three operations preserve the fact that the cost function is a piecewise quadratic function. All three operations preserve the property that, for every point $t$ on the piecewise quadratic function, the left derivative at $t$ is greater than or equal to the right derivative at $t$.
\end{proof}

\begin{lemma}
\label{lem:type_B_propagate_in_constant_time}
Let $A_{k-1,\ell'}$ be the parent of $A_{k,\ell}$. Suppose all paths from $A_{k-1,\ell'}$ are type $(B)$ paths. If we propagate a single quadratic piece of $A_{k-1,\ell'}$ to the next level $A_k$, it is a piecewise quadratic function with at most a constant number of pieces. Moreover, this propagation step takes only constant time.
\end{lemma}

\begin{proof}
Computing the cost function along $V$ consists of adding a quadratic function to the cost function along $A_{k-1,\ell}$, and taking the cumulative minimum. This cost function has at most two pieces and can be computed in constant time. Adding a quadratic function to the cost function along $V$ to yield the cost function along $A_k$ takes constant time.
\end{proof}

\begin{lemma}
\label{lem:type_B_bounds}
Let $A_{k-1,\ell'}$ be the parent of $A_{k,\ell}$. Suppose all paths from $A_{k-1,\ell'}$ to $A_{k,\ell}$ are type $(B)$ paths. Then $$D(A_{k,\ell}) \leq D(A_{k-1,\ell'}) + 1,$$ $$|A_{k,\ell}| \leq |A_{k-1,\ell'}| + D(A_{k-1,\ell'}) + 1.$$
\end{lemma}

\begin{proof}
Three operations are performed to get from the cost function along $A_{k-1,\ell'}$ to the cost function along $A_{k,\ell}$. The first and third are to add a quadratic function, whereas the second is to take the cumulative minimum.

The first and third operations do not change the number of distinct pairs $(a,b)$. The second operation only adds horizontal functions, so the number of distinct pairs $(a,b)$ increases by at most once. Putting this together yields $D(A_{k,\ell}) \leq D(A_{k-1,\ell'}) + 1$. 

The first and third operations do not change the number of pieces in the piecewise quadratic function. The second operation adds a horizontal function at each of the local minima. The local minima of a piecewise quadratic function can occur at the endpoints of the function, at points where both the left and right derivatives are zero, or at points where the left derivative is strictly less than the right derivative. There are at most two endpoints, and we can only add a horizontal function to one of them. For each distinct pair $(a,b)$, there is at most one local minima where both the left and right derivatives are zero, since this local minima only occurs at the $x$-coordinate given by $x = -\frac b {2a}$. By Lemma~\ref{lem:type_B_piecewise_quadratic}, we do not have any points where the left derivative is strictly less than the right derivative. Putting this together, we add at most $D(A_{k-1,\ell'}) + 1$ horizontal functions at the local minima of the piecewise quadratic function. This gives $|A_{k,\ell}| \leq |A_{k-1,\ell'}| + D(A_{k-1,\ell'}) + 1$, as required.
\end{proof}

\subsection{Type \textit{(C)} paths}

Recall that we require the following three lemmas for type $(C)$ paths.

\begin{lemma}
\label{claim:type_C_piecewise_quadratic}
Let $A_{k-1,\ell'}$ be the parent of $A_{k,\ell}$. Suppose all paths from $A_{k-1,\ell'}$ to $A_{k,\ell}$ are type $(C)$ paths. Then the cost along $A_{k,\ell}$ is a continuous piecewise quadratic function. Moreover, for every point $t$ on the piecewise quadratic function, the left derivative at $t$ is greater than or equal to the right derivative at $t$.
\end{lemma}

\begin{lemma}
\label{claim:type_C_propagate_in_constant_time}
Let $A_{k-1,\ell'}$ be the parent of $A_{k,\ell}$. Suppose all paths from $A_{k-1,\ell'}$ are type $(C)$ paths. If we propagate a single quadratic piece of $A_{k-1,\ell'}$ to the next level $A_k$, it is a piecewise quadratic function with at most a constant number of pieces. Moreover, this propagation step takes only constant time.
\end{lemma}

\begin{lemma}
\label{claim:type_C_bounds}
Let $A_{k-1,\ell'}$ be the parent of $A_{k,\ell}$. Suppose all paths from $A_{k-1,\ell'}$ to $A_{k,\ell}$ are type $(C)$ paths. Then $$D(A_{k,\ell}) \leq D(A_{k-1,\ell'}),$$ $$|A_{k,\ell}| \leq |A_{k-1,\ell'}|.$$ 
\end{lemma}

We prove these lemmas again by dividing into subcases. Our subcases are defined as follows.

\begin{definition}
Let type $(C1)$, $(C2)$ and $(C3)$ paths be type $(C)$ paths where the starting point and the ending point are on the left and right, bottom and right, or bottom and top boundaries respectively.
\end{definition}

\begin{figure}[ht]
    \centering
    \includegraphics{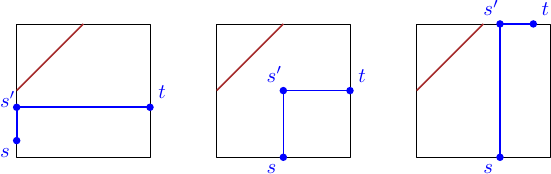}
    \caption{The type $(C1)$, $(C2)$ and $(C3)$ paths.}
    \label{fig:c1c2c3}
\end{figure}

To show that we can indeed divide our analysis into these three subcases, we make the following observation.

\begin{observation}
Let $A_{k-1,\ell'}$ be the parent of $A_{k,\ell}$. Suppose there exists a type $(C1)$, or $(C2)$, or $(C3)$ path from $A_{k-1,\ell'}$ to $A_{k,\ell}$. Then all paths from $A_{k-1,\ell'}$ to $A_{k,\ell}$ are type $(C1)$, or $(C2)$, or $(C3)$, respectively.
\end{observation}

\begin{proof}
Suppose there were paths of different types between $A_{k-1,\ell'}$ to $A_{k,\ell}$. Then either their starting points are on different boundaries, or their ending points are on different boundaries. Therefore, there is either a critical point on $A_{k-1,\ell'}$ or on $A_{k,\ell}$. But this is not possible, so there are no paths of different types between $A_{k-1,\ell'}$ to $A_{k,\ell}$.
\end{proof}

We require separate case analyses for type $(C1)$, $(C2)$ and $(C3)$ paths. See Figure~\ref{table:c_divided_c1c2c3}. We prove our claims for type $(C1)$ paths in Section~\ref{sec:c1}, type $(C2)$ paths in Section~\ref{sec:c2}, and type $(C3)$ paths in Section~\ref{sec:c3}.

\begin{figure}[ht]
    \centering
    \begin{tabular}{|l|l|l|l|}
        \hline
         & Type $(C1)$ & Type $(C2)$ & Type $(C3)$ \\
        \hline
        Lemma~\ref{claim:type_C_piecewise_quadratic}
        & Lemma~\ref{lem:type_C1_piecewise_quadratic}
        & Lemma~\ref{lem:type_C2_piecewise_quadratic}
        & Lemma~\ref{lem:type_C3_piecewise_quadratic}
        \\
        \hline
        Lemma~\ref{claim:type_C_propagate_in_constant_time}
        & Lemma~\ref{lem:type_C1_propagate_in_constant_time}
        & Lemma~\ref{lem:type_C2_propagate_in_constant_time}
        & Lemma~\ref{lem:type_C3_propagate_in_constant_time}
        \\
        \hline
        Lemma~\ref{claim:type_C_bounds}
        & Lemma~\ref{lem:type_C1_bounds}
        & Lemma~\ref{lem:type_C2_bounds}
        & Lemma~\ref{lem:type_C3_bounds}
        \\
        \hline
    \end{tabular}
    \caption{The proofs of Claims~\ref{claim:type_C_piecewise_quadratic}, \ref{claim:type_C_propagate_in_constant_time} and~\ref{claim:type_C_bounds} divided into three subcases.}
    \label{table:c_divided_c1c2c3}
\end{figure}

\subsection{Type (C1) paths}
\label{sec:c1}

The type $(C1)$ path is first a vertical path from $s$ to $s'$, then a horizontal path from $s'$ to~$t$. However, we notice that since $s'$ is on the left boundary, we can simplify the path to the horizontal part from $s'$ to $t$, since we know the optimal cost at $s'$ by our inductive hypothesis. The cost along the horizontal path from $s'$ to $t$ is a quadratic, so the cost function along $A_{k,\ell}$ is simply the cost function along $A_{k-1,\ell'}$ plus a quadratic. Using the same arguments as in Section~\ref{sec:case_analysis_type_a_paths}, we yield the following three lemmas.

\begin{lemma}
\label{lem:type_C1_piecewise_quadratic}
Let $A_{k-1,\ell'}$ be the parent of $A_{k,\ell}$. Suppose all paths from $A_{k-1,\ell'}$ to $A_{k,\ell}$ are type $(C1)$ paths. Then the cost along $A_{k,\ell}$ is a continuous piecewise quadratic function. Moreover, for every point $t$ on the piecewise quadratic function, the left derivative at $t$ is greater than or equal to the right derivative at $t$.
\end{lemma}

\begin{lemma}
\label{lem:type_C1_propagate_in_constant_time}
Let $A_{k-1,\ell'}$ be the parent of $A_{k,\ell}$. Suppose all paths from $A_{k-1,\ell'}$ are type $(C1)$ paths. If we propagate a single quadratic piece of $A_{k-1,\ell'}$ to the next level $A_k$, it is a piecewise quadratic function with at most a constant number of pieces. Moreover, this propagation step takes only constant time.
\end{lemma}

\begin{lemma}
\label{lem:type_C1_bounds}
Let $A_{k-1,\ell'}$ be the parent of $A_{k,\ell}$. Suppose all paths from $A_{k-1,\ell'}$ to $A_{k,\ell}$ are type $(C1)$ paths. Then $$D(A_{k,\ell}) \leq D(A_{k-1,\ell'}),$$ $$|A_{k,\ell}| \leq |A_{k-1,\ell'}|.$$ 
\end{lemma}

\subsection{Type (C2) paths}
\label{sec:c2}

\begin{definition}
Let $A_{k-1,\ell'}$ be the parent of $A_{k,\ell}$. Suppose all paths from $A_{k-1,\ell'}$ to $A_{k,\ell}$ are type $(C2)$ paths. Let the type $(C2)$ path be a vertical path from $s$ to $s'$ and a horizontal path from $s'$ to $t$. Define $pathcost(s,t)$ to be the cost at $s$ plus the cost along the type $(C2)$ path. In other words,

\[
pathcost(s,t) = cost(s) + \int_{s}^{s'} h(x) \cdot dx + \int_{s'}^t h(x) \cdot dx.
\] 
\end{definition}

The cost at $t$ is the minimum path cost over all type $(C2)$ paths ending at~$t$. In other words, $cost(t) = \min_{s} pathcost(s,t)$.

\begin{observation}
\label{obs:c2_boundaries_not_optimal}
Suppose $s$ is a point on $A_{k-1,\ell'}$ where the left derivative of $cost(s)$ at $s$ does not match the right derivative of $cost(s)$ at $s$. Then $cost(t) \neq pathcost(s,t)$ for all $t$.
\end{observation}

\begin{proof}
We assume by inductive hypothesis that Lemma~\ref{claim:details_piecewise_quadratic} is true along $A_{k-1,\ell'}$. In particular, we assume the cost function along $A_{k-1,\ell'}$ cannot obtain a local minimum at a point where its left derivative does not equal its right derivative. Recall that $pathcost(s,t) = cost(s) + \int_{s}^{s'} h(x) \cdot dx + \int_{s'}^t h(x)$. For fixed $t$, the second and third terms are (single piece) quadratic functions in terms of $s$. Therefore, for fixed $t$, we add two single piece quadratic functions to $cost(s)$ to obtain $pathcost(s,t)$. Therefore, $pathcost(s,t)$ cannot have a local minimum at $s$, since the left derivative does not match the right derivative at $s$ in $cost(s)$, and adding the two quadratic functions to $cost(s)$ does not affect this property. Since $pathcost(s,t)$ does not have a local minimum at $s$, it cannot have a global minimum at $s$ either, so $cost(t) \neq pathcost(s,t)$, as required.
\end{proof}

\begin{observation}
\label{obs:partial_derivative}
Let $A_{k-1,\ell'}$ be the parent of $A_{k,\ell}$. Suppose that $s$ is on $A_{k-1,\ell'}$ and $t$ is on $A_{k,\ell}$ so that $cost(t) = pathcost(s,t)$. Then $\frac {\partial} {\partial s} pathcost(s,t) = 0$. Moreover, the set of points in a cell for which this occurs is a set of segments, and there are at most $|A_{k-1,\ell'}|$ such segments.
\end{observation}


\begin{proof}
By the contrapositive of Observation~\ref{obs:c2_boundaries_not_optimal}, we have that the point $s$ is a point on $A_{k-1,\ell'}$ such that its the left derivative of $cost(s)$ at $s$ matches the right derivative of $cost(s)$ at $s$. So the derivative of $cost(s)$ is well defined at $s$. Therefore, $\frac {\partial} {\partial s} pathcost(s,t)$ is well defined at $s$, since $pathcost(s,t) = cost(s) + \int_{s}^{s'} h(x) \cdot dx + \int_{s'}^t h(x)$, and the second and third terms are quadratic functions in terms of $s$ if $t$ is fixed. Since $\frac {\partial} {\partial s} pathcost(s,t)$ is well defined at $s$, and $cost(t) = pathcost(s,t)$ so $pathcost(s,t)$ is minimised as $s$, we have that  $\frac {\partial} {\partial s} pathcost(s,t) = 0$. Since $pathcost(s,t)$ is a bivariate quadratic function with $|A_{k-1,\ell'}|$ pieces, the partial derivative equation $\frac {\partial} {\partial s} pathcost(s,t) = 0$ defines a bivariate linear function with $|A_{k-1,\ell'}|$ pieces. This bivariate linear function is a set of $|A_{k-1,\ell'}|$ segments, as required.
\end{proof}

\begin{lemma}
\label{lem:type_C2_piecewise_quadratic}
Let $A_{k-1,\ell'}$ be the parent of $A_{k,\ell}$. Suppose all paths from $A_{k-1,\ell'}$ to $A_{k,\ell}$ are type $(C2)$ paths. Then the cost along $A_{k,\ell}$ is a continuous piecewise quadratic function. Moreover, for every point $t$ on the piecewise quadratic function, the left derivative at $t$ is greater than or equal to the right derivative at $t$.
\end{lemma}

\begin{proof}
For all $t$ on $A_{k,\ell}$, we define $\sigma(t)$ as the point on $A_{k-1,\ell'}$ such that $cost(t) = pc(t)$ where~$pc(t) := pathcost(\sigma(t),t) = cost(\sigma(t)) + \int_{\sigma(t)}^{\sigma'(t)} h(r) dr + \int_{\sigma'(t)}^t h(r) dr$. By induction, we get that~$cost$ is a piecewise quadratic function on~$A_{k-1,\ell'}$, and by Observation~\ref{obs:partial_derivative}, we get that $\sigma$ is a piecewise linear function. Therefore,~$pc$ is a piecewise quadratic function and so is~$cost$ on~$A_{k,l}$.

Next, we show that for every point $t_0$ on $A_{k,\ell}$, $pc$ is continuous in~$t_0$ (which isn't necessarily true for $\sigma$) and the left derivative of~$pc$ at $t_0$ is greater than or equal to the right derivative of~$pc$ at~$t_0$. We write~$\partial^- pc(t_0) \geq \partial^+ pc(t_0)$. Consider the points directly to the left of~$t_0$ such that a single linear piece of~$\sigma$ contains all of them in its domain. Now let~$\overrightarrow \sigma$ be the linear function corresponding to this piece with its domain extended to the right of~$t_0$, and let~$\overrightarrow{pc}$ be the function~$\overrightarrow{pc}(t) := pathcost (\overrightarrow \sigma (t),t) = cost(\overrightarrow \sigma (t)) + \int_{\overrightarrow \sigma(t)}^{\overrightarrow \sigma'(t)} h(r) dr + \int_{\overrightarrow \sigma'(t)}^t h(r) dr$. We analogously define $\overleftarrow \sigma$ and $\overleftarrow{pc}$ by considering points directly to the right of~$t_0$ and extending the piece's domain to the left. Note that $\overrightarrow{\sigma} \neq \overleftarrow{\sigma}$ only if $\sigma$ changes from one linear piece to another at $t_0$.

Both~$\overrightarrow{pc}$ and $\overleftarrow{pc}$ are continuous, since~$\overrightarrow \sigma$ and~$\overleftarrow \sigma$ are linear while $cost$ is continuous on~$A_{k-1,l'}$ by induction. For all~$t < t_0$ and close to~$t_0$, we have~$\overrightarrow{pc}(t) = pc(t) = cost(t) = \min_s pathcost(s,t)$ by definition and thus~$\overrightarrow{pc}(t) \leq \overleftarrow{pc}(t)$. Similarly, for all~$t > t_0$ close to~$t_0$, we have~$\overrightarrow{pc}(t) \geq \overleftarrow{pc}(t)$. Together, this implies~$\overrightarrow{pc}(t_0) = \overleftarrow{pc}(t_0)$, which in turn yields that~$pc$ is continuous in~$t_0$ with $pc(t_0) = \overrightarrow{pc}(t_0) = \overleftarrow{pc}(t_0)$. 

The continuity of $\overrightarrow{pc}$, $\overleftarrow{pc}$ and $pc$ at $t_0$ also means that they all have a left and a right derivative at~$t_0$. Assume for the sake of contradiction that~$\partial^-\overrightarrow{pc}(t_0) \neq \partial^+\overrightarrow{pc}(t_0)$. This would require $\partial^- cost(\overrightarrow{\sigma}(t_0)) \neq \partial^+ cost(\overrightarrow{\sigma}(t_0))$, because the second and third terms of $\overrightarrow{pc}(t)$ are quadratic functions of $t$ without break points. Due to the linearity of~$\overrightarrow \sigma$, $\partial^- cost(\overrightarrow \sigma(t_0)) \neq \partial^+ cost(\overrightarrow \sigma(t_0))$ only occurs when $\overrightarrow \sigma(t_0)$ is a break point of $cost$, which is impossible by~$\overrightarrow{pc}(t_0) = pc(t_0) = cost(t_0)$ and the contrapositive of Observation~\ref{obs:c2_boundaries_not_optimal}. This yields a contradiction, so~$\partial^-\overrightarrow{pc}(t_0) = \partial^+\overrightarrow{pc}(t_0)$.

Finally, we have $\partial^-\overrightarrow{pc}(t_0) = \partial^-{pc}(t_0)$ and $\partial^+\overrightarrow{pc}(t_0) \geq \partial^+{pc}(t_0)$, since $\overrightarrow{pc}(t) = pc(t)$ for all~$t \leq t_0$ and $\overrightarrow{pc}(t) \geq pc(t)$ for all $t > t_0$ close to~$t_0$. Putting everything together yields
\[
\partial^-cost(t_0) 
= \partial^-{pc}(t_0)
= \partial^-\overrightarrow{pc}(t_0)
= \partial^+\overrightarrow{pc}(t_0)
\geq \partial^+{pc}(t_0)
 = \partial^+cost(t_0)
\]
as required.
\end{proof}

\begin{lemma}
\label{lem:type_C2_propagate_in_constant_time}
Let $A_{k-1,\ell'}$ be the parent of $A_{k,\ell}$. Suppose all paths from $A_{k-1,\ell'}$ are type $(C2)$ paths. If we propagate a single quadratic piece of $A_{k-1,\ell'}$ to the next level $A_k$, the cost is a quadratic function. Moreover, this propagation step takes only constant time.
\end{lemma}

\begin{proof}
Let $s$ be a point on a single quadratic piece on $A_{k-1,\ell'}$. In constant time, we can compute the bivariate quadratic function $pathcost(s,t)$. In constant time, we can compute the segment where $\frac \partial {\partial s} pathcost(s,t) = 0$. For all $(s,t)$ on this segment, we write $s$ as a linear function of $t$. We know by Observation~\ref{obs:partial_derivative} that $cost(t) = pathcost(s,t)$ along this segment. In constant time, we can substitute the linear function of $s$ in terms of $t$ into $pathcost(s,t)$ to obtain $cost(t)$, i.e. the propagated cost along the next level $A_k$.
\end{proof}

\begin{lemma}
\label{lem:type_C2_bounds}
Let $A_{k-1,\ell'}$ be the parent of $A_{k,\ell}$. Suppose all paths from $A_{k-1,\ell'}$ to $A_{k,\ell}$ are type $(C2)$ paths. Then $$D(A_{k,\ell}) \leq D(A_{k-1,\ell'}),$$ $$|A_{k,\ell}| \leq |A_{k-1,\ell'}|.$$ 
\end{lemma}

\begin{proof}
First, we show $|A_{k,\ell}| \leq |A_{k-1,\ell'}|.$ For all $t$ on $A_{k,\ell}$, we define $s(t)$ as the point on $A_{k-1,\ell'}$ such that that $cost(t) = pathcost(s(t),t)$. By Observation~\ref{obs:partial_derivative}, we get that $s(t)$ is a piecewise linear function in terms of $t$, and $s(t)$ has at most $|A_{k-1,\ell'}|$ linear pieces. Substituting $s(t)$ into $pathcost(s(t),t)$, we get that $cost(t)$ is a piecewise quadratic function, which is the lower envelope of $|A_{k-1,\ell'}|$ quadratic pieces. By Lemma~\ref{lem:crossing_lemma_main_paper}, the lower envelope of $|A_{k-1,\ell'}|$ pieces appear in the same order as the input pieces they were propagated from. Therefore, the lower envelope $cost(t)$ has at most $|A_{k-1,\ell'}|$ quadratic pieces, as required.

Next, we show $D(A_{k,\ell}) \leq D(A_{k-1,\ell'})$. Suppose that $f_1(x) = ax^2 + bx + c_1$, and $f_2(x) = ax^2 + bx + c_2$. We apply the propagation step in Lemma~\ref{lem:type_C2_propagate_in_constant_time} to $f_1(x)$ and $f_2(x)$. Let $pathcost_1(s,t)$ be the bivariate quadratic function for $s$ in the domain of $f_1(x)$, and similarly $pathcost_2(s,t)$ be the bivariate quadratic function for $s$ in the domain of $f_2(x)$. Since $f_1(x)$ and $f_2(x)$ differ by at most a constant, we get that $pathcost_1(s,t)$ and $pathcost_2(s,t)$ differ by at most a constant. The derivative $\frac \partial {\partial s} pathcost(s,t)$ does not depend on the constant terms, so $s$ and $t$ have the same linear relationship for $s$ in the domain of $f_1(x)$ and $f_2(x)$. Substituting the linear function $s(t)$ into $pathcost_1(s(t),t)$ and $pathcost_2(s(t),t)$, we obtain $cost_1(t)$ and $cost_2(t)$, that differ at most by an additive constant. Hence, any two quadratic functions in $cost(s)$ that share the same $(a,b)$ terms propagate to a pair of quadratic functions in $cost(t)$ that share the same $(a,b)$ terms. Hence, every one of the $D(A_{k-1,\ell'})$ many distinct $(a,b)$ pairs in $cost(s)$ map to at most one of the $D(A_{k,\ell})$ many distinct $(a,b)$ pairs in $cost(t)$. Therefore, $D(A_{k,\ell}) \leq D(A_{k-1,\ell'})$, as required.
\end{proof}

\subsection{Type (C3) paths}
\label{sec:c3}

\begin{definition}
Let $A_{k-1,\ell'}$ be the parent of $A_{k,\ell}$. Let the type $(C3.1)$ function be the cost function along $A_{k,\ell}$ if we only allow vertical type $(C3)$ paths from $A_{k-1,\ell'}$ to $A_{k,\ell}$. Let the type $(C3.2)$ function be the cost function along $A_{k,\ell}$ if we allow arbitrary type $(C3)$ paths (vertical followed by a horizontal path) from $A_{k-1,\ell'}$ to $A_{k,\ell}$. See Figure~\ref{fig:type_c3.1_c3.2}.
\end{definition}

\begin{figure}[ht]
    \centering
    \includegraphics{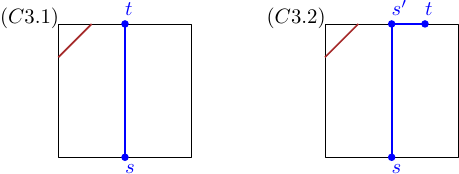}
    \caption{The type $(C3.1)$ and type $(C3.2)$ paths.}
    \label{fig:type_c3.1_c3.2}
\end{figure}

Similar to type $(B1)$ and type $(B2)$ functions in Section~\ref{sec:b}, our approach will be to first compute the type $(C3.1)$ and then modify it to obtain the type $(C3.2)$ function. Recall that for a type $(B1)$ function $f(t)$, the type $(B2)$ function is the cumulative minimum of $f(t)$. Recall that the cumulative minimum has horizontal extensions at each of the local minima of $f(t)$. For type $(C3)$ paths, we apply similar extensions, but our new extensions will be non-horizontal, but instead follow the shape of a quadratic function $g(t)$. See Figure~\ref{fig:continous_continuation}. 

\begin{figure}[ht]
    \centering
    \includegraphics{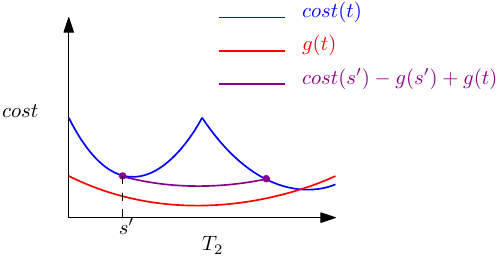}
    \caption{A non-horizontal extension (purple) of a function (blue) that follows the shape of a quadratic (red).}
    \label{fig:continous_continuation}
\end{figure}

\begin{observation}
\label{obs:c3.2_to_c3.1}
Let $cost_{3.1}(t)$ be a type $(C3.1)$ function, and let $g(t)$ be the integral of the height function $h(t)$ along the top boundary. Then the type $(C3.2)$ function, $cost_{3.2}(t)$, is given by
\[
        cost_{3.2}(t) = \min_{s' \leq t} \big( cost_{3.1}(s') - g(s') \big) + g(t) 
\]

\end{observation}

\begin{proof}
Recall that a type $(C3.2)$ path is a path from $s$ (bottom) to $s'$ (top) to $t$ (top). The type $(C3.2)$ cost function is the minimum cost over all such type $(C3.2)$ paths. Therefore, 
\begin{align*}
    cost_{3.2}(t) &= \min_{s' \leq t} \big(cost(s) + \int_s^{s'} h(t) dt + \int_{s'}^{t} h(t) dt \big) \\
    &= \min_{s' \leq t} \big(cost_{3.1}(s') + \int_{s'}^{t} h(t) dt \big) \\
    &= \min_{s' \leq t} \big(cost_{3.1}(s') - g(s') + g(t) \big) \\
    &= \min_{s' \leq t} \big(cost_{3.1}(s') - g(s')\big) + g(t),
\end{align*}
as required.
\end{proof}

Now, we can compute $cost_{3.2}(t)$ from $cost(s)$ in the following way. First, we compute the quadratic function $\int_{s}^{s'} h(t) dt$ and add it to $cost(s)$ to obtain $cost_{3.1}(s')$. Next, we compute $g(s')$, the integral of the height function along the top boundary, and subtract it from $cost_{3.1}(s')$. Then, we compute its cumulative minimum function, $\min_{s' \leq t} \big( cost_{3.1}(s') - g(s') \big)$. Finally, we add $g(t)$. In each step, we either add or subtract a quadratic, or take the cumulative minimum of a function. These steps are essentially the same as the ones for type $(B)$ paths, with minor modifications. By applying essentially the same arguments as in Section~\ref{sec:b}, we obtain the following three lemmas.

\begin{lemma}
\label{lem:type_C3_piecewise_quadratic}
Let $A_{k-1,\ell'}$ be the parent of $A_{k,\ell}$. Suppose all paths from $A_{k-1,\ell'}$ to $A_{k,\ell}$ are type $(C3)$ paths. Then the cost along $A_{k,\ell}$ is a continuous piecewise quadratic function. Moreover, the local minima along $A_{k,\ell}$ are either at its endpoints or where its derivative is well defined and equal to zero.
\end{lemma}

\begin{lemma}
\label{lem:type_C3_propagate_in_constant_time}
Let $A_{k-1,\ell'}$ be the parent of $A_{k,\ell}$. Suppose all paths from $A_{k-1,\ell'}$ are type $(C3)$ paths. If we propagate a single quadratic piece of $A_{k-1,\ell'}$ to the next level $A_k$, it is a piecewise quadratic function with at most a constant number of pieces. Moreover, this propagation step takes only constant time.
\end{lemma}

\begin{lemma}
\label{lem:type_C3_bounds}
Let $A_{k-1,\ell'}$ be the parent of $A_{k,\ell}$. Suppose all paths from $A_{k-1,\ell'}$ to $A_{k,\ell}$ are type $(C3)$ paths. Then $$D(A_{k,\ell}) \leq D(A_{k-1,\ell'}),$$ $$|A_{k,\ell}| \leq |A_{k-1,\ell'}|.$$ 
\end{lemma}

\section{Conclusion}

We presented the first exact algorithm for computing CDTW of one-dimensional curves, which runs in polynomial time. 
Our main technical contribution is bounding the total complexity of the functions which the algorithm propagates, to bound the total running time of the algorithm. One direction for future work is to improve the upper bound on the total complexity of the propagated functions. Our $O(n^5)$ upper bound is pessimistic, for example, we do not know of a worst case instance. Another direction is to compute CDTW in higher dimensions. In two dimensions, the Euclidean $\mathcal L_2$ norm is the most commonly used norm, however, this is likely to result in algebraic issues similar to that for the weighted region problem~\cite{DBLP:journals/comgeo/CarufelGMOS14}. One way to avoid these algebraic issues is to use a polyhedral norm, such as the $\mathcal L_1$, $\mathcal L_\infty$, or an approximation of the $\mathcal L_2$ norm~\cite{dudley1974metric,har2019proof}. This would result in an approximation algorithm similar to~\cite{DBLP:journals/comgeo/MaheshwariSS18}, but without a dependency on the spread. 

\subsection*{Acknowledgements}

The authors would like to thank Jan Erik Swiadek for helpful feedback on the manuscript. In particular, we included their improved proof of Lemma~\ref{lem:type_C2_piecewise_quadratic}.

\bibliographystyle{plain}
\bibliography{main.bib}

\newpage
\appendix

\section{Proof of Lemma~\ref{lem:cdtw_formula}}
\label{apx:proof_of_cdtw_formula}

\cdtwformula*

\begin{proof}
Recall that the definition of $d_{CDTW}(P,Q)$ is:
\[
d_{CDTW}(P,Q) = 
        \inf_{(\alpha,\beta) \in \Gamma(p,q)} \int_0^1 ||P(\alpha(z)) - Q(\beta(z))|| \cdot ||\alpha'(z) + \beta'(z)|| \cdot dz \\
\]

Let $\gamma = (\alpha,\beta) \in \Gamma(p,q)$. Then $\gamma(0) = (\alpha(0), \beta(0)) = (0,0)$ and $\gamma(1) = (\alpha(1),\beta(1)) = (p,q)$. By considering $\gamma(z)$ as a point in the the parameter space $R$, we observe that if we vary $z \in [0,1]$, then $\gamma(z)$ is a curve starting at $(0,0)$, ending at $(p,q)$, and is non-decreasing in both $x$- and $y$-coordinates. Now, consider the integral of $h(\cdot)$ along the curve $\gamma$. The mathematical expansion of the line integral $\int_\gamma h(z) \cdot dz$ is:

\[
    \int_\gamma h(z) \cdot dz = \int_0^1 h(\gamma(z)) \cdot ||\gamma'(z)||_R \cdot dz.
\]

The expanded integral closely resembles the formula for $d_{CDTW}(P,Q)$. The first term $h(\gamma(z))$ is equal to $||P(\alpha(z)) - Q(\beta(z))||$. The second term $||\gamma'(z)||_R$ is equal to $||\alpha'(z) + \beta'(z)||$, since $||\cdot||_R$ is the $\mathcal L_1$ norm and because $\alpha$ and $\beta$ are non-decreasing. Hence,
\[
d_{CDTW}(P,Q)
    =
        \inf_{\gamma \in \Gamma(p,q)} \int_0^1 h(\gamma(z)) \cdot ||\gamma'(z)||_R \cdot dz. \\
\]

We now reparametrise $\gamma \subset R$ in terms of its $\mathcal L_1$ arc length in $R$. This is the ``natural parametrisation'' of the curve $\gamma$. We already know that $\gamma$ starts at $(0,0)$, ends at $(p,q)$, and is differentiable and  non-decreasing in $x$- and $y$-coordinates. We let $\Psi(p,q)$ be the set of all functions that satisfy these three conditions, in addition to a fourth condition, $||\gamma'(z)||_R = 1$. Then $\alpha'(z) + \beta'(z) = 1$, as $\alpha'(z), \beta'(z) > 0$. By integrating from $0$ to $z$, we get $\alpha(z) + \beta(z) = z$. In particular, we have $\gamma(p+q) = (p,q)$, since the curve must end at $(p,q)$. So $z \in [0,p+q]$ is the new domain of the reparametrised curve $\gamma \in \Psi(p,q)$. Putting this together, we obtain the stated lemma.
\end{proof}

\end{document}